\documentclass[aps,pra,11pt,onecolumn,nofootinbib,tightenlines]{revtex4-2}

\usepackage{mathrsfs}
\usepackage{graphicx}
\usepackage{comment}
\usepackage{dsfont}
\usepackage{amsmath}
\usepackage{amsfonts}
\usepackage{amssymb}
\usepackage{braket}
\usepackage{amsthm}
\usepackage[colorlinks=true,urlcolor=blue,citecolor=blue,linkcolor=blue]{hyperref}
\usepackage{caption}
\usepackage{lipsum}
\usepackage{subcaption}
\usepackage{mwe}
\usepackage{enumerate}
\usepackage{sidecap}
\usepackage{natbib}
\usepackage{appendix}
\usepackage{multirow}
\usepackage{array}
\usepackage{physics}
\usepackage{mathtools}
\usepackage{booktabs}

\usepackage[paperwidth=215mm,paperheight=297mm,centering,hmargin=2.45cm,vmargin=2.5cm]{geometry}

\newtheorem{theorem}{Theorem}
\newtheorem{lemma}[theorem]{Lemma}
\newtheorem{proposition}[theorem]{Proposition}
\newtheorem{corollary}[theorem]{Corollary}
\theoremstyle{definition}
\newtheorem{remark}{Remark}

\newcommand{\fo}{\rightarrowtail}

\newcommand{\PPT}{\text{PPT}}
\newcommand{\qstate}{\mathcal{S}(A)}

\DeclareMathOperator*{\argmin}{arg\,min}

\renewcommand{\tilde}[1]{\widetilde{#1}} 
\renewcommand{\hat}[1]{\widehat{#1}}

\def \d {\mathrm{d}}

\DeclareMathOperator{\sinc}{sinc}

\DeclareMathOperator{\SEP}{SEP}
\DeclareMathOperator{\supp}{supp}

\makeatletter

\renewcommand*{\thesubsection}{\thesection.\arabic{subsection}}
\renewcommand*{\p@subsection}{}

\renewcommand*{\p@subsubsection}{}
\makeatother

\begin{document}

\title{New additivity properties of the relative entropy of entanglement and its generalizations} 

\begin{abstract}
We prove that the relative entropy of entanglement is additive when \emph{at least one of the two states} belongs to some specific class. We show that these classes include bipartite pure, maximally correlated, GHZ, Bell diagonal, isotropic, and generalized Dicke states. Previously, additivity was established only if \textit{both} states belong to the same class. Moreover, we extend these results to entanglement monotones based on the $\alpha$-$z$ R\'enyi relative entropy. Notably, this family of monotones includes also the generalized robustness of entanglement and the geometric measure of entanglement. In addition, we prove that any monotone based on a quantum relative entropy is not additive for general states. 
We also compute closed-form expressions of the monotones for bipartite pure, Bell diagonal, isotropic, generalized Werner, generalized Dicke, and maximally correlated Bell diagonal states. Our results rely on developing a method that allows us to recast the initial convex optimization problem into a simpler linear one. Even though we mostly focus on entanglement theory, we expect that some of our technical results could be useful in investigating more general convex optimization problems. 
\end{abstract}

\author{Roberto Rubboli}
\email{roberto.rubboli@u.nus.edu}
\affiliation{Centre for Quantum Technologies, National University of Singapore, Singapore 117543, Singapore}

\author{Marco Tomamichel}
\affiliation{Department of Electrical and Computer Engineering,
National University of Singapore, Singapore 117583, Singapore}
\affiliation{Centre for Quantum Technologies, National University of Singapore, Singapore 117543, Singapore}

\maketitle

\section{Introduction}
Quantifying entanglement is a central problem in quantum information theory. Entanglement monotones have the desirable property that they do not increase under local operations and classical communications (LOCC), so-called free operations in the resource theory of entanglement. Furthermore, additive entanglement monotones remain monotones even in the presence of catalysts that can be leveraged in the procedure but need to be returned unchanged. Establishing the additivity properties of an entanglement monotone is also of particular relevance in the analysis of transformation rates between states under LOCC where we expect operationally relevant quantities to be (at least) weakly additive. For example, it is shown that the regularized entropy of entanglement provides an upper bound on transformation rates~\cite{horodecki2001entanglement}. More in general, in convex resource theories, the regularized relative entropy of resource gives an upper bound on state transformations under free operations~\cite{horodecki2013quantumness}. Thus, a fundamental problem is to establish whether a certain entanglement monotone or measure is additive and, if not, what are the minimum constraints on the states to ensure its additivity. While this has been an active area for many decades, a complete understanding of the additivity properties of many entanglement monotones is still missing.

In this work, we focus on the additivity properties of the relative entropy of entanglement~\cite{vedral1997quantifying}. The relative entropy of entanglement is an example of a monotone that is known to be not additive under tensor product~\cite{vollbrecht2001entanglement}. Nevertheless, we show that additivity holds whenever one of the two states belongs to some specific class. This proves for the first time that for several classes of states, only one state needs to be of some specific form to ensure its additivity. Previously, additivity was proven only in the case where both states belong to some specific class, e.g. when both states are bipartite pure~\cite{vedral1998entanglement} and left as an open question for the case where only one state is constrained. Hence, whether the relative entropy of entanglement is additive when only one state is of some specific form was still not known.

We also mention that in the seminal work~\cite{rains1999bound}, the author proved that the relative entropy optimized over the set of
states with positive partial transpose (REEP) is additive for the tensor product of two maximally
correlated states. Moreover, when one state commutes with its optimizer, under some further
specific conditions, additivity results can be obtained for the REEP (see Section~\ref{First class} for a more detailed discussion).
In addition, in~\cite{miranowicz2008closed} the authors proved that the REEP is weakly additive for any two-qubit state that commutes with its optimizer. The latter result relies on a closed formula for the inverse problem of finding all the entangled states of a given closest separable state. This result was later generalized for all dimensions and number of parties in~\cite{friedland2011explicit} and in~\cite{girard2014convex} to the Rains bound and other functions of interest in quantum information theory.

We then generalize our results to the monotones based on the $\alpha$-$z$  R\'enyi relative entropies which provide a general framework to address different families of entanglement monotones. Notably, this family includes, among many others, also the geometric measure~\cite{barnum2001monotones} and the generalized robustness of entanglement~\cite{Datta_rob2}. 

These additivity results have direct application in deriving fundamental limits of quantum tasks, where no dimensional or purity assumption can be made on one of the two states. An example is provided by catalytic resource theories. Here, the additivity properties of the monotones based on the $\alpha$-$z$ R\'enyi relative entropies are tightly connected with their role in characterizing catalytic state conversions~\cite{rubboli2021fundamental}. In this setting, different monotones for different values of the parameters have different roles. Therefore, it is crucial to completely characterize the additivity properties of these monotones for all the values of the parameters. An important open question is whether a catalyst in a mixed state could provide an advantage over pure state catalysts~\cite{datta2022catalysis}. Our results immediately provide a set of necessary conditions for catalytic transformation without correlation of pure entangled states with mixed catalyst (see Section~\ref{Bipartite pure section} for more details). This set of conditions is a proper subset of the necessary (and sufficient) one for pure state catalysts~\cite{Klimesh}. However, whether these conditions are also sufficient is still an open question. Moreover, for correlated catalysis, i.e. when there are residual correlations between the system and the catalyst at the end of the process, the currently known protocols use catalysts which, for small residual correlations, are typically highly dimensional~\cite{Kondra, Lipka}. Our results allow us to establish fundamental limits on correlated catalytic transformations, which hold for any protocol and do not require any assumption on the structure of the catalyst. Explicitly, we are interested in the additivity of the monotones for $\alpha=z$ with $\alpha \in [1/2,1)$ of a pure state with any catalyst state~\cite{rubboli2021fundamental}. The latter requirement on the range of the parameters highlights the importance of extending the result outside the relative entropy of entanglement case. 

Entanglement monotones based on quantum relative entropies play a pivotal role in entanglement theory. Some notable examples include the above-mentioned relative entropy of entanglement, generalized robustness, and geometric measure of entanglement. A fundamental question is whether an entanglement monotone based on a quantum relative entropy is additive or not. Even though it was already known that the relative entropy of entanglement and the generalized robustness are not additive for general states~\cite{vollbrecht2001entanglement,zhu2010additivity}, a complete answer to this question was still missing. We answer this question by showing that this is not the case, i.e. any entanglement monotone based on a quantum relative entropy is not additive for general states.

Optimization problems involving quantum relative entropies are ubiquitous in quantum information theory. Some examples include the optimization of the fidelity function or the Umegaki relative entropy over a set of states satisfying some specific constraints. Here, we develop a very general ansatz-based technique for a wide range of quantum relative entropies that allows us to recast the initial convex optimization problem into a simpler linear one. We believe that these conditions are very general and could potentially be applied to a wide range of (even classical) convex optimization problems. Our results hinge on this method which is based on deriving necessary and sufficient conditions for the optimizer of the monotones based on the $\alpha$-$z$  R\'enyi relative entropies. We show that these conditions provide a powerful tool to investigate additivity questions for any resource theory by first considering the exemplary case of the resource theory of coherence. In this case, we can prove
that the monotones based on the $\alpha$-$z$ R\'enyi divergences are additive for any pair of states by extending some known results (see the discussion in Section~\ref{Additivity coherence}). In particular, we can show that an optimizer of the tensor product of two states is the tensor product of optimizers of the marginal problems. More interestingly, in entanglement theory, we can show that
the optimizer of the tensor product of two states still factorizes when one state belongs to some specific class.

In the case when the state commutes with its optimizer, we find that these necessary and sufficient conditions considerably simplify. This allows us to readily provide a counterexample to the additivity of any entanglement monotone based on a quantum relative entropy. Moreover, these conditions allow us to analytically compute the monotones in a simple fashion for several states by extending and unifying some already known results (see Section~\ref{Analytics} and Table~\ref{states}).

The paper is structured as follows:
\begin{itemize}
\item In Section~\ref{monotones} we introduce the monotones based on $\alpha$-$z$ R\'enyi relative entropies, and we discuss their connections with the relative entropy of entanglement, the geometric measure, and the generalized robustness of entanglement. 

\item In Section~\ref{condition} we derive necessary and sufficient conditions for the optimizer of very general convex optimization problems involving the $\alpha$-$z$ R\'enyi relative entropies. 

\item In Section~\ref{Additivity} we prove new additivity properties of the entanglement monotones based on $\alpha$-$z$ R\'enyi relative entropies. We prove that they are additive whenever one state belongs to one of two classes. The first class includes the states that commute with their optimizer such that the alpha power of their product with the inverse of their optimizer has positive entries on a product basis. The second class is the set of maximally correlated states, which, notably, contains the bipartite pure states. Finally, we extend the latter result to also include the multipartite GHZ state.  

\item In Section~\ref{Counterexample} we prove that any monotone based on a quantum relative entropy is not additive for general states by providing a counterexample.  Moreover, we prove that they are not additive also when the minimization over the separable states is replaced by the minimization over the states with positive partial transpose (PPT).

\item Finally, in Section~\ref{Analytics} we provide some examples of states that belong to the additivity classes mentioned above. In particular, we show that maximally correlated, bipartite pure, GHZ, Bell diagonal, isotropic, generalized Dicke,  and separable states belong to these classes. Moreover, we compute their value for bipartite pure, Bell diagonal, isotropic, generalized Werner, generalized Dicke, and maximally correlated Bell diagonal states by generalizing some known results (see Table~\ref{states}).

\end{itemize}

\section{Resource monotones based on $\alpha$-$z$ R\'enyi relative entropies} 
\label{monotones}

We denote with $\mathcal{P}(A)$ the set of positive operators on a Hilbert space $A$. Moreover, we denote with $\qstate$ the set of quantum states, i.e. the subset of $\mathcal{P}(A)$ with unit trace. Let $\sigma \in \mathcal{S}(\otimes_{j=1}^N A_j)$ be a $N$-partite state shared among N parties with Hilbert spaces $A_1, \hdots ,A_N$, respectively. We say that $\sigma$ is separable if it is of the form $\sigma = \sum_i p_i \sigma_i^1 \otimes \hdots \otimes \sigma_i^N $ for some local states $\sigma_i^j \in \mathcal{S}(A_j)$ and a probability distribution $\{ p_i \}$. Otherwise, we call it entangled. We denote the set of all separable states by $\SEP(A_1:\hdots:A_N)$ or just $\SEP$ if the geometry is clear from the context. The extreme points of the set of separable states are the pure product states $\sigma =\ketbra{\sigma}{\sigma}$ with $|\sigma \rangle = |\sigma_1\rangle \otimes \hdots\otimes  |\sigma_N\rangle$. We denote the set of pure product states with $\text{PRO}(A_1:\hdots:A_N)$ or simply $\text{PRO}$.
 
A key property that a measure of entanglement is required to satisfy is the monotonicity under local operations and classical communications (LOCC)~\cite{horodecki2009quantum,vedral1998entanglement,vedral1997quantifying}.
We call a function $\mathfrak{R} : \mathcal{S}(A) \rightarrow [0, + \infty]$ an entanglement monotone if it does not increase under LOCC, i.e., if $\mathfrak{R}(\rho) \geq \mathfrak{R}(\mathcal{E}(\rho))$ for any state $\rho$ and any LOCC operation $\mathcal{E}$~\cite{vidal2000entanglement}. For an entanglement monotone to be called an entanglement measure, further properties such as faithfulness, convexity, or full monotonicity are usually required~\cite{horodecki2009quantum,vedral1998entanglement,vedral1997quantifying}.  

Here, we follow the standard convention where, given two $N$-partite states  $\rho_1\in \mathcal{S}(\otimes_{j=1}^N A_j)$  and  $\rho_2 \in \mathcal{S}(\otimes_{j=1}^N A'_j)$, the parties $A_j$ and $A'_j$ for $j=1,\cdots,N$ are assumed to be in the same laboratory and hence the tensor product $\rho_1 \otimes \rho_2$ is considered to be a $N$-partite state shared among the parties $A_1A'_1, \hdots ,A_N A'_N$. Similarly, we extend this convention to tensor products of more than two states. We say that $\mathfrak{R}$ is \textit{tensor sub-additive} (or just sub-additive) for the states $\rho_1$ and $\rho_2$ if $\mathfrak{R}(\rho_1 \otimes \rho_2) \leq \mathfrak{R}(\rho_1) +  \mathfrak{R}(\rho_2)$.  Moreover, we say that $\mathfrak{R}$ is \textit{tensor additive} (or just additive) for the states $\rho_1$ and $\rho_2$ if $\mathfrak{R}(\rho_1 \otimes \rho_2) = \mathfrak{R}(\rho_1) +  \mathfrak{R}(\rho_2)$.

We now introduce the entanglement monotones based on the $\alpha$-$z$ R\'enyi relative entropy. We remark that even though we formulate our results very generally for the monotones based on the $\alpha$-$z$ R\'enyi relative entropy, our results are also new for the special cases of most interest, namely the relative entropy of entanglement, the generalized robustness of entanglement and the geometric measure of entanglement (see the discussion below for more details).  Let $ \alpha \in (0,1) \cup(1,\infty), \; z>0$, $\rho \in \qstate$ and $ \sigma \in \mathcal{P}(A)$. Then the $\alpha$-$z$ \textit{R\'enyi relative entropy} of $\sigma$ with $\rho$ is defined as~\cite{audenaert2015alpha,zhang2020wigner}
\begin{equation}
\label{definition}
D_{\alpha,z}(\rho \| \sigma):=
\begin{cases}
\frac{1}{\alpha-1}\log{\text{Tr}\left(\rho^\frac{\alpha}{2z}\sigma^\frac{1-\alpha}{z}\rho^\frac{\alpha}{2z}\right)^z} & \text{if}\; (\alpha<1 \wedge \rho \not \perp \sigma) \vee \rho \ll \sigma \\
 +\infty & \text{else}
\end{cases} \,.
\end{equation}
In the following, we denote $Q_{\alpha,z}(\rho \| \sigma):=\exp \big( (\alpha-1)D_{\alpha,z}(\rho \| \sigma) \big)$.
In the limit points of the ranges of the parameters, we define the $\alpha$-$z$ R\'enyi relative entropy by taking the corresponding pointwise limits.
In particular, we have the pointwise limits~\cite{audenaert2015alpha,lin2015investigating}
\begin{equation}
\label{limits}
\quad D_{\min}(\rho \| \sigma) = \lim \limits_{\alpha \rightarrow 0} D_{\alpha,1}(\rho \| \sigma) \;, \quad D(\rho \| \sigma) = \lim \limits_{\alpha \rightarrow 1} D_{\alpha,\alpha}(\rho \| \sigma) \;, \quad D_{\max}(\rho \| \sigma) = \lim \limits_{\alpha \rightarrow \infty} D_{\alpha,\alpha}(\rho \| \sigma)   \,,
\end{equation}
where 
\begin{align}
\label{Dmin}
& D_{\min}(\rho \| \sigma) := -\log{\Tr(\Pi(\rho)\sigma)} \, , \\
&D(\rho \| \sigma) := \Tr(\rho (\log{\rho} - \log{\sigma}))\,, \,  \text{and}\\
&D_{\max}(\rho \| \sigma) := \inf \{ \lambda \in \mathbb{R} : \rho \leq 2^{\lambda}\sigma \} \, ,
\end{align}
are the \textit{min-relative entropy}~\cite{Renner,Datta_rob2}, the Umegaki relative entropy, and the  \textit{max-relative entropy}~\cite{Tomamichel,Datta_rob2, Renner}, respectively. Note that the second limit in~\eqref{limits} for the Umegaki relative entropy actually holds for any $z > 0$~\cite{lin2015investigating}. Here, we denoted with $\Pi(\rho)$ the projector onto the support of $\rho$.

%Moreover, in the limit $z \rightarrow \infty$, we have for $\alpha \neq 1$~\cite{audenaert2015alpha}
%\begin{equation}
%\lim \limits_{z \rightarrow \infty} D_{\alpha,z}(\rho \| \sigma) =  \frac{1}{\alpha-1}\log{\Tr \exp{(\alpha \log{\rho}+(1-\alpha)\log{\sigma})}}\,,
%\end{equation}
%where $\log{\rho}$ and $\log{\sigma}$ need to be restricted to the intersection of the supports of $\rho$ and $\sigma$.In the limit $\alpha \rightarrow 1$, the above expression also converges to the Umegaki relative entropy~\cite{audenaert2015alpha,gour2020optimal}.
   
When $z=1$, the $\alpha$-$z$ R\'enyi relative entropy reduces to the \emph{Petz R\'enyi divergence}~\cite{petz1986quasi,Tomamichel}
\begin{equation}
\bar{D}_{\alpha}(\rho \| \sigma) :=
\frac{1}{\alpha-1} \log{\Tr(\rho^\alpha \sigma^{1-\alpha})} \,.
\end{equation}
For $z=\alpha$, it reduces to the \textit{sandwiched quantum R\'enyi divergence}~\cite{Muller, Wilde3,Tomamichel},
\begin{equation}
\tilde{D}_{\alpha}(\rho \| \sigma) :=
\frac{1}{\alpha-1}\log{\Tr\big(\sigma^{\frac{1-\alpha}{2\alpha}}\rho \sigma^{\frac{1-\alpha}{2\alpha}}\big)^{\alpha}} \,.
\end{equation}
For $z=1-\alpha$, we find the \textit{reverse sandwiched R\'enyi relative entropy}~\cite{audenaert2015alpha},
\begin{equation}
\tilde{D}^{\text{rev}}_{\alpha}(\rho \| \sigma) :=
\frac{1}{\alpha-1}\log{\Tr\big(\rho^{\frac{\alpha}{2(1-\alpha)}}\sigma \rho^{\frac{\alpha}{2(1-\alpha)}}\big)^{1-\alpha}} \,.
\end{equation}

In~\cite[Theorem 1.2]{zhang2020wigner} the authors proved that the $\alpha$-$z$ R\'enyi relative entropy satisfies the data-processing inequality (DPI) if and only if one of the following holds,
\begin{enumerate}
\item  $\quad 0<\alpha<1 \; \; and \; \; z\geq \max\{\alpha,1-\alpha\},$ 
\item $\quad 1<\alpha \leq 2 \;\;  and \;\;  \frac{\alpha}{2} \leq z \leq \alpha,\; and $
\item $\quad 2 \leq \alpha<\infty \;\;  and \;\;  \alpha-1 \leq z \leq \alpha\, . $
\end{enumerate}
We denote with $\mathcal{D}$ the set of above values of the parameters $(\alpha,z)$ for which the $D_{\alpha,z}$ satisfies the DPI.
Because of the limits in~\eqref{limits}, the $\alpha$-$z$ R\'enyi relative entropy also satisfies the DPI on the line $\alpha=1$. We include this line in the region $\mathcal{D}$.
%Moreover, being limits of quantities that satisfy the DPI, it satisfies the DPI in the points $\alpha=0 \wedge z = 1$ and  $z=\infty \wedge \alpha \leq 1$. We include these points in the region $\mathcal{D}$.
In Fig.~\ref{fig: alpha-z} we represent the limits discussed above in the $\alpha$-$z$ plane and we show in blue the region for which the DPI holds. 
 
\begin{figure}
\centering
{\includegraphics[width=1.\textwidth]{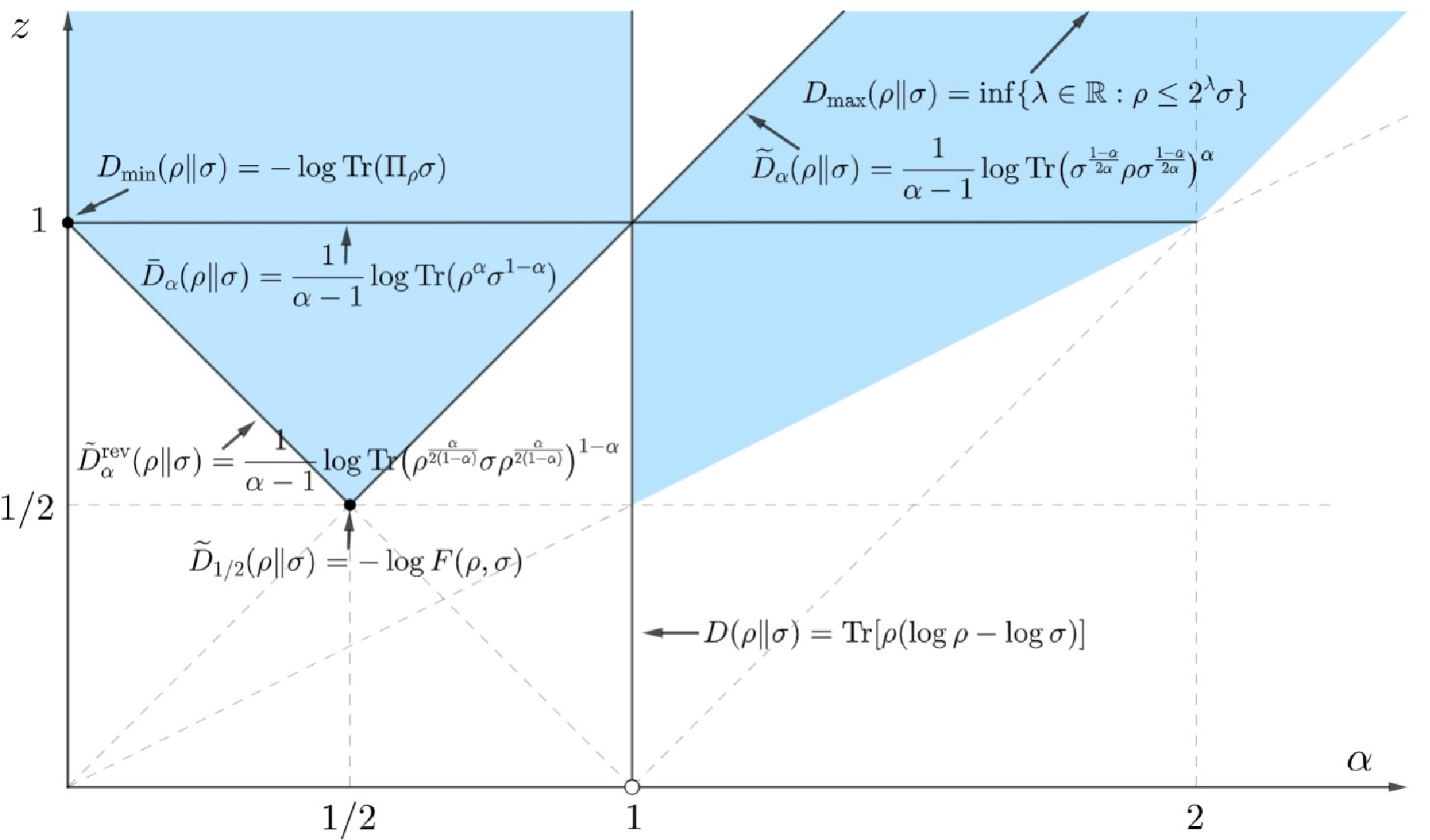}}
\caption{$\alpha$-$z$ plane for the $\alpha$-$z$ R\'enyi relative entropies. The $z=1$ and the $z=\alpha$ lines correspond to the Petz and the sandwiched R\'enyi divergences, respectively. The vertical line $\alpha=1$ for $z > 0$ corresponds to the Umegaki relative entropy. The blue region (together with the Umegaki relative entropy line) is where the data-processing inequality holds}
\label{fig: alpha-z}
\end{figure}

We define entanglement monotones for $(\alpha,z)\in\mathcal{D}$ as
\begin{align}
\label{problem}
&\mathfrak{D}_{\alpha,z} (\rho) := \inf_{\sigma \in \SEP} D_{\alpha,z}(\rho \| \sigma) \,.
\end{align}
In particular,
\begin{equation}
\mathfrak{D}_{\min}(\rho) :=   \inf_{\sigma \in \SEP} D_{\min}(\rho \| \sigma) , \quad \mathfrak{D}(\rho) :=   \inf_{\sigma \in \SEP} D(\rho \| \sigma) , \quad \mathfrak{D}_{\max}(\rho):= \inf_{\sigma \in \SEP} D_{\max}(\rho \| \sigma) \,.
\end{equation}
In Appendix~\ref{lower-semicontinuous} we show that the latter monotones coincide with the corresponding pointwise limits of the monotone~\eqref{problem}. Namely, we have
\begin{equation}
\label{limits monotone}
\mathfrak{D}_{\min}(\rho) = \lim \limits_{\alpha \rightarrow 0}\mathfrak{D}_{\alpha,1}(\rho) , \quad \mathfrak{D}(\rho) = \lim \limits_{\alpha \rightarrow 1}\mathfrak{D}_{\alpha,\alpha}(\rho) , \quad \mathfrak{D}_{\max}(\rho) = \lim \limits_{\alpha \rightarrow \infty}\mathfrak{D}_{\alpha,\alpha}(\rho) \,.
\end{equation} 
In the following, we refer to the $\mathfrak{D}_{\alpha,z}$ as $\alpha$-$z$ R\'enyi relative entropy of entanglement or as entanglement monotone based on $\alpha$-$z$ R\'enyi relative entropy. In Appendix~\ref{lower-semicontinuous} we show that the above infimum is always achieved, i.e. there always exists a closest separable state and hence the infimum can be replaced by the minimum. However, in general, the minimum might not be unique as shown in~\cite[Example 21]{lami2021attainability} for $\alpha=z=1$.

It is straightforward to prove that, for the values of $\alpha$ and $z$ for which the underlying $\alpha$-$z$ R\'enyi relative entropy satisfies the data-processing inequality, $\mathfrak{D}_{\alpha,z}$ does not increase under LOCC, i.e. it is indeed an entanglement monotone.  Moreover, in this interval, the function $Q_{\alpha,z}(\rho\| \sigma)$ is jointly concave for $\alpha \leq 1$ and jointly convex for $\alpha \geq 1$~\cite{zhang2020wigner}. This implies that any local minimum of the problem~\eqref{problem} is also a global minimum. This fact has been already pointed out for $\alpha=z=1$ in~\cite{vedral1998entanglement}.

In~\cite{vedral1998entanglement}, the authors propose a numerical method to compute the relative entropy of entanglement for bipartite states. Following this approach, the minimum of the relative entropy of entanglement can be found with a gradient search in multidimensional parameter space. This is done using Caratheodory's theorem, which provides a way to bound the number of independent real parameters needed to parametrize a separable state. The same strategy could be used to find the minimum of the problem~\eqref{problem}. However, this method becomes soon impractical even for low dimensions. More recently, a method to compute a lower bound on the relative entropy of entanglement using semidefinite program solvers was presented in~\cite{fawzi2018efficient}.

Interestingly, the $\mathfrak{D}_{\alpha,z}$ encompass several well-known entanglement measures. The monotone $\mathfrak{D}$ is, by definition, equal to the relative entropy of entanglement. The relative entropy of entanglement was the first divergence-based entanglement measure to be studied and was introduced in~\cite{vedral1997quantifying}. 

In~\cite{Datta_rob2} it was shown that the generalized (log) robustness is equal to the entanglement monotone based on the $D_{\max}$ divergence, i.e. $\mathfrak{D}_{\max}(\rho)=\log{(1 +\mathfrak{R}_g(\rho))}$ where
\begin{equation}
\mathfrak{R}_g(\rho) := \min \left\{ s\geq 0: \exists \omega \in \mathcal{S}(A) \,\, \text{s.t} \,\, \frac{1}{1+s}\left(\rho + s\omega\right) \in \SEP  \right\}\, .
\end{equation} 
is the generalized robustness of entanglement. Historically, the robustness of entanglement was introduced first in~\cite{vidal1999robustness}.  The robustness of entanglement quantifies how much mixing with a separable state (separable noise) can take place before an entangled state becomes a separable state. Subsequently, in~\cite{Steiner} and~\cite{harrow2003robustness} the authors introduced the generalized robustness of entanglement. The latter quantity generalizes the robustness of entanglement, and it quantifies how much mixing with a state (not necessarily separable) can take place before an entangled state becomes separable. 

The monotone $\mathfrak{D}_{1/2,1/2}(\rho)$ is linked to the fidelity of separability and the geometric measure of entanglement through the relations $\mathfrak{D}_{1/2,1/2}(\rho) = -\log{F_s(\rho)}=-\log{(1-E_G(\rho))}$. The fidelity of separability is $F_s(\rho)=\max_{\sigma \in \SEP}F(\rho,\sigma)$.  Here, $F(\rho, \sigma) :=  ( \Tr|\sqrt{\rho}\sqrt{\sigma}| )^2$ is the Uhlmann's fidelity. The geometric measure of entanglement is defined through the convex-roof construction as follows:
\begin{align}
&E_G(|\psi\rangle) = 1- \max_{|\phi\rangle \in \text{PRO}}|\langle\phi|\psi\rangle|^2 \, ,\\
&E_G(\rho) = \min_{\{p_i,|\psi_i\rangle\}} \sum_i p_i E_G(|\psi_i\rangle)\,,
\end{align} 
where the minimization is performed over all the decomposition $\rho = \sum_i p_i |\psi_i\rangle\!\langle \psi_i|$ into pure states.  The geometric measure of entanglement was first proposed in~\cite{barnum2001monotones} and subsequently investigated in~\cite{wei2003geometric,chen2010computation,hubener2009geometric}.  The proof of the above link was provided in~\cite{streltsov2010linking} where the authors showed the key relation $E_G(\rho) = 1 - F_s(\rho)$. Note that the additivity of $\mathfrak{D}_{1/2,1/2}$ is equivalent to the multiplicativity of the fidelity of separability $F_s$. 

More recently, the limits $z=1$ and $z=\alpha$ have been investigated in~\cite{zhu2017coherence}. In these cases, the $\mathfrak{D}_{\alpha,z}$ reduces to the entanglement monotone based on the Petz  R\'enyi divergence and to the one based on the sandwiched R\'enyi divergence, respectively.

We introduce the Fr\'echet derivative~(see e.g. \cite[Sec. V.3 and Sec. X.4]{bhatia1997graduate} ). We use the shorthand $\partial_x:=\partial/\partial x$. For $\rho, \sigma, \tau \in\mathcal{S}(\mathcal{H})$ we define
\begin{equation}
\left.\partial_\sigma D_{\alpha,z}(\rho\| \tau):=  \partial_xD_{\alpha,z}(\rho\| (1-x)\tau+x\sigma)\right|_{x=0} \,.
\end{equation}
In the following, we use the notational convention $A/B:=B^{-\frac{1}{2}}AB^{-\frac{1}{2}}$ for Hermitian operators $A$ and $B$. Moreover, we define the negative powers in the sense of generalized inverses; i.e., when we take the negative powers of positive operators, we simply ignore the kernel.

\section{Necessary and sufficient conditions for the optimizer} 
\label{condition}
In this section, we derive necessary and sufficient conditions for the optimal state of the minimization problem~\eqref{problem}. 
We formulate our result for very general convex optimization problems. Let $\mathcal{F}$ be a closed and convex set that contains a state with full support. We define
\begin{align}
\label{problem F}
&\mathfrak{D}^\mathcal{F}_{\alpha,z} (\rho) := \inf_{\sigma \in \mathcal{F}} D_{\alpha,z}(\rho \| \sigma) \,.
\end{align}
Note that the above minimum is always achieved (see Appendix~\ref{lower-semicontinuous}). In quantum resource theories, the above function quantifies the amount of `quantum resource' contained in a state (see~\cite{Gour,Brandao2} for a review). 
Resource theories offer a general framework to quantify the usefulness of quantum states and their inconvertibility under free operations. 
A convex resource theory is defined by a closed convex subset of quantum
states called \textit{free states} and a set of \textit{free operations} with the property that they are closed under composition
and map free states into free states. Hence, in this setting, $\mathcal{F}$ is the set of free states. Some examples of resource theories include entanglement theory~\cite{horodecki2009quantum,Plenio,vedral1997quantifying}, athermality in thermodynamics~\cite{Brandao3,faist2015gibbs,Oppenheim}, and coherence~\cite{Winter,aberg2006quantifying,baumgratz2014quantifying}.  
We call a function a resource monotone if it does not increase under free operations. It is straightforward to prove that the quantity~\eqref{problem F} is a resource monotone for $(\alpha,z) \in \mathcal{D}$.
In entanglement theory, where $\mathcal{F}= \SEP$, the above quantity reduces to the entanglement monotone in~\eqref{problem}. In quantum resource theories, the existence of a full-rank state is equivalent to the requirement that the monotone~\eqref{problem F} must always be finite for all $\alpha \geq 1$ and full support states. The latter requirement is a reasonable assumption since in any reasonable resource theory, strictly speaking, no state can have infinite resources. 
However, we remark that the results we derive in this section are very general and the set $\mathcal{F}$ does not necessarily have to be linked to any quantum resource theory. We first prove two auxiliary lemmas. 
\begin{lemma}
\label{technical}
Let $\tau > 0$ and $\sigma \geq 0$. Then, for $\beta \in  (-1,0) \cup (0,1) $, we have
\begin{equation}
 \partial_x((1-x)\tau + x \sigma)^\beta  =  \frac{\sin(\pi \beta)}{\pi} \int_0^\infty \frac{\sigma-\tau}{((1-x)\tau + x \sigma+t)^2} t^\beta \d t  \,.
\end{equation}
In particular, we have
\begin{equation}
\label{eq lemma}
\left. \partial_x((1-x)\tau + x \sigma)^\beta \right|_{x=0} =  \frac{\sin(\pi \beta)}{\pi} \int_0^\infty \frac{\sigma}{(\tau+t)^2} t^\beta \d t - \beta \tau^\beta \,.
\end{equation}
\end{lemma}
\begin{proof}
For a positive definite operator $A>0$ we have the following integral formulas
\begin{align}
& A^\beta = -\frac{\sin{(\pi \beta)}}{\pi}\int_0^\infty \frac{1}{A+t}t^\beta \d t \quad \text{for all} \quad \beta \in (-1,0)\, ,\;  \text{and}\\
& A^\beta = \frac{\sin{(\pi \beta)}}{\pi} \int_0^\infty \frac{A}{A+t} t^{\beta-1} \d t  \quad \text{for all} \quad \beta \in (0,1) \, .
\end{align}
We then use the above formulas to obtain
\begin{align}
\label{first 1}
& \partial_x((1-x)\tau + x \sigma)^\beta = -\frac{\sin{(\pi \beta)}}{\pi}\int_0^\infty  \partial_x \frac{1}{(1-x)\tau + x \sigma+t}t^\beta \d t \quad \beta \in (-1,0) \,, \;  \text{and}\\
\label{second 1}
&  \partial_x((1-x)\tau + x \sigma)^\beta  = \frac{\sin{(\pi \beta)}}{\pi} \int_0^\infty   \partial_x   \frac{(1-x)\tau + x \sigma}{(1-x)\tau + x \sigma+t}      t^{\beta-1} \d t   \quad \beta \in (0,1) \,.
\end{align}
Moreover, we use that $ \partial \sigma(x)^{-1}/\partial x = - \sigma(x)^{-1} (\partial \sigma(x)/\partial x)  \sigma(x)^{-1} $ to obtain
\begin{align} 
\label{first 2}
&  \partial_x\frac{1}{(1-x)\tau + x \sigma+t} = - \frac{\sigma-\tau}{((1-x)\tau + x \sigma + t)^2} \,,\;  \text{and}\\
\label{second 2}
&  \partial_x \frac{(1-x)\tau + x \sigma}{(1-x)\tau + x \sigma+t} =  t \frac{\sigma-\tau}{((1-x)\tau + x \sigma + t)^2} \,.
\end{align} 
Combing equations~\eqref{first 1} and~\eqref{second 1} together with equations~\eqref{first 2} and~\eqref{second 2} proves the first part of the lemma. We now evaluate the derivative at $x=0$. We obtain
\begin{align}
\left. \partial_x ((1-x)\tau + x \sigma)^\beta \right|_{x=0} & = \frac{\sin(\pi \beta)}{\pi} \int_0^\infty \frac{\sigma-\tau}{(\tau+t)^2} t^\beta \d t \\
& = \frac{\sin(\pi \beta)}{\pi} \int_0^\infty \frac{\sigma}{(\tau+t)^2} t^\beta \d t - \frac{\sin(\pi \beta)}{\pi} \tau \int_0^\infty \frac{1}{(\tau+t)^{2}} t^\beta \d t \,.
\end{align}
We now use that, for a positive operator $A>0$, we have 
\begin{equation}
 \int_0^\infty \frac{1}{(A+ t)^2} t^{\beta} \d t = \frac{A^{\beta-1}}{\sinc(\pi \beta) }  \quad \text{for all} \quad \beta \in (-1,0) \cup (0,1) \,,
\end{equation}
to integrate the second term in the r.h.s of the previous expression.
This proves equation~\eqref{eq lemma}. 
\end{proof}
\begin{lemma}
\label{support}
Let $A,B \geq 0$. The following two conditions are equivalent:
\begin{enumerate}[ {(}1{)} ]
\item $\exists\;  0 < \gamma,\beta < \infty$ such that $\gamma \Pi(A) \leq ABA \leq \beta \Pi(A)$.
\item $\textup{supp}(ABA)=\textup{supp}(A)$.
\end{enumerate}
\end{lemma}
\begin{proof}
We always have $\supp(ABA)\subseteq \supp(A)$ and $ABA \leq \beta \Pi(A)$ with $\beta = \lambda_{\max}(B) \lambda^2_{\max}(A)$. Here, for a positive operator $A$, we denoted with $\lambda_{\max}(A)$ the maximum eigenvalue of $A$.

We first show $\mathit{(1)} \!\! \implies \!\! \mathit{(2)}$. Let us choose $|\psi\rangle \in \supp(A)$. We have
\begin{equation}
\langle \psi|ABA|\psi\rangle \geq \gamma \langle \psi| \Pi(A) |\psi\rangle \neq 0 \, .
\end{equation} 
which implies that $\supp(ABA)\supseteq \supp(A)$ and hence $\supp(ABA) = \supp(A)$.

We now prove that $\mathit{(2)} \!\! \implies \!\! \mathit{(1)}$. For a positive operator $A$, we denoted with $\lambda_{\min}(A)$ the minimum non-zero eigenvalue of $A$. We have
\begin{equation}
ABA \geq \lambda_{\min}(ABA) \Pi(ABA) = \lambda_{\min}(ABA)\Pi(A) \, .
\end{equation}
The implication follows by setting $\gamma=\lambda_{\min}(ABA)$.
\end{proof}

We now state the necessary and sufficient conditions for the optimizer(s) of~\eqref{problem F}. To properly formulate the necessary and sufficient conditions of the optimizer of the problem~\eqref{problem F} (see Theorem~\ref{main Theorem} below), we introduce the following set of states. For any $\rho \in \qstate$ and $(\alpha,z)\in \mathcal{D}$, we define
\begin{align}
\label{sets}
S_{\alpha,z}(\rho) =  \begin{cases} 
\left\{ \sigma \in \qstate : \supp(\rho) = \supp\!\left(\Pi(\rho)\sigma \Pi(\rho)\right) \right\}  &  \textup{if} \; (1-\alpha)/z = 1\\
\left\{ \sigma \in \qstate : \supp(\rho) \subseteq \supp(\sigma)  \right\} &  \textup{otherwise} 
\end{cases} \, ,
\end{align} 
where we denoted with $\Pi(\rho)$ the projector onto the support of $\rho$ as before. 
Note that the line $z=1-\alpha$ is located at the boundary of the DPI region for $\alpha \leq 1/2$ (see Fig.~\ref{fig: alpha-z}). We now show that the optimizer(s) of~\eqref{problem F} must belong to the previous set. We remark that these support conditions are very relevant, as without them the theorem could lead to wrong conclusions.  We have the following, 
\begin{lemma}
\label{necessary conditions}
Let $\rho$ be a quantum state and $(\alpha,z) \in \mathcal{D}$. If $\tau \in \argmin_{\sigma \in \mathcal{F}} D_{\alpha,z}(\rho \| \sigma)$ then $\tau \in S_{\alpha,z}(\rho)$.
\end{lemma}
We split the proof for the region $(\alpha,z) \in \mathcal{D}$ in different ranges.
\begin{proof}[Proof for $\alpha \geq 1$]
For $\alpha \geq 1$, if  $\tau \in \argmin_{\sigma \in \mathcal{F}} D_{\alpha,z}(\rho \| \sigma)$, it must be that $\supp(\rho)\subseteq \supp(\tau)$ or otherwise the $\mathfrak{D}_{\alpha,z}(\rho)$, by definition, would be infinite. However, by assumption, there always exists a free state $\sigma$ with full support such that $\mathfrak{D}_{\alpha,z}(\rho) \leq D_{\alpha,z}(\rho \| \sigma) < \infty$. 
\end{proof}

\begin{proof}[Proof for $z=1-\alpha \wedge \alpha \in (0,1)$]
In this case, we show that a necessary condition for $\tau$ to be an optimizer is that $ \supp\!\left(\Pi(\rho)\tau \Pi(\rho)\right) = \supp(\rho)$. To do that, we show that, if $\tau$ does not satisfy the latter conditions on the support, we can always find a sufficiently small $x$ and a free state $\sigma $ such that $D_{\alpha,z}(\rho \| (1-x)\tau+x\sigma) < D_{\alpha,z}(\rho \| \tau)$ (or, since  $\alpha \in (0,1)$, $Q_{\alpha,z}(\rho \| (1-x)\tau+x\sigma) > Q_{\alpha,z}(\rho \| \tau)$). This contradicts the fact that $\tau$ is an optimizer. Here, we implicitly used that the resource theory is convex, i.e. $(1-x)\tau+x\sigma \in \mathcal{F}$.

We start by calculating the Fr\'echet derivative along a free direction $\sigma$ with full support. We have
\begin{align}
\partial_x Q_{\alpha,z}(\rho \| (1-x)\tau + x \sigma) =& z\Tr(\chi_{\alpha,z}(\rho,(1-x)\tau + x \sigma) (\sigma-\tau)) \\
\label{derivative}
=& z\Tr(\chi_{\alpha,z}(\rho,(1-x)\tau + x \sigma) \sigma)-z\Tr(\chi_{\alpha,z}(\rho,(1-x)\tau + x \sigma) \tau) \, .
\end{align}
where we defined the positive operator $\chi_{\alpha,z}(\rho,\sigma):=  \rho^\frac{\alpha}{2z}( \rho^{\frac{\alpha}{2z}}\sigma^{\frac{1-\alpha}{z}} \rho^{\frac{\alpha}{2z}})^{z-1} \rho^\frac{\alpha}{2z}$. The above function is a continuous and finite function of $x$ for $x \neq 0$. 
In the following, we repeatedly use that if $A,A',B \geq 0$ and $A \leq A'$, then $ \Tr(AB)\leq \Tr(A'B)$.
We now bound separately the two terms of equation~\eqref{derivative}. We define $\beta:= \alpha/(2z)$. To bound the first term, we use that $\rho^{\beta}((1-x)\tau + x \sigma)\rho^{\beta} \leq (1-x)\Pi(\rho^{\beta} \tau  \rho^{\beta}) + x \Pi(\rho)$. Here, we used that $\rho^{2 \beta} \leq \Pi(\rho)$ and that $\Tr(\tau \rho^{2\beta}) \leq 1$. This implies that
\begin{align}
\label{1st line}
\Tr(\chi_{\alpha,z}(\rho,(1-x)\tau + x \sigma) \sigma) &= \Tr\left(\rho^{\beta}(\rho^{\beta}((1-x)\tau + x \sigma)\rho^{\beta})^{-\alpha}\rho^{\beta}\sigma\right)  \\
\label{1st implication}
&\geq \lambda_{\min}(\sigma) \Tr\left(\rho^{2\beta}(\Pi(\rho^{\beta} \tau  \rho^{\beta})+x^{-\alpha}(\Pi(\rho)-\Pi(\rho^{\beta} \tau  \rho^{\beta}))\right)\\
\label{first term derivative I}
& \geq  x^{-\alpha} \Tr\left(\Pi(\rho)(\Pi(\rho)-\Pi(\rho^\beta \tau \rho^\beta))\right) K \,,
\end{align}
where we denoted with $\lambda_{\min}(\rho)$ the minimum non-zero eigenvalue of $\rho$, and we defined $K:=\lambda_{\min}(\sigma)\lambda^{2\beta}_{\min}(\rho)$. In~\eqref{1st implication} we used that $t \mapsto t^s$ is operator antimonotone for $s \in [-1,0)$. 
We now bound the second term in~\eqref{derivative}. We first note that $Q_{\alpha,z}(\rho \| \sigma) \leq 1$ for any $\rho, \sigma \in \qstate$. Indeed, we have
\begin{equation}
Q_{\alpha,z}(\rho \| \sigma) = \int Q_{\alpha,z}\left(U\rho U^\dagger \middle\| U\sigma U^\dagger \right) \d U \leq Q_{\alpha,z}\left( \int U\rho U^\dagger \d U \middle\| \int U\sigma U^\dagger \d U \right)  = Q_{\alpha,z}\left(\frac{\mathds{1}}{d} \middle\| \frac{\mathds{1}}{d} \right) =1 \,,
\end{equation} 
where we denoted with $\d U$ the invariant Haar measure. In the first equality, we used the unitary invariance of $Q_{\alpha,z}$ and in the inequality, we used that for $\alpha \in (0,1)$, the $Q_{\alpha,z}$ is jointly concave. 

We then use that $\tau \leq \tau +(x/(1-x))\sigma$ to get
\begin{align}
 \Tr(\chi_{\alpha,z}(\rho,(1-x)\tau + x \sigma) \tau) & = \Tr\left(\rho^{\beta}(\rho^{\beta}((1-x)\tau + x \sigma)\rho^{\beta})^{-\alpha}\rho^{\beta}\tau\right) \\
&\leq \frac{1}{1-x} Q_{\alpha,z}(\rho \| (1-x)\tau + x \sigma) \\
& \leq 2 \,.
\end{align}
In the last inequality, we used that $Q_{\alpha,z}(\rho \| \sigma) \leq 1$ and we assumed that $x \leq 1/2$. Combining the above results we obtain, for $x \leq 1/2$,
\begin{equation}
\label{lower bound derivative}
\partial_x Q_{\alpha,1-\alpha}(\rho \| (1-x)\tau + x \sigma) \geq  x^{-\alpha} \Tr\left(\Pi(\rho)(\Pi(\rho)-\Pi(\rho^\beta \tau \rho^\beta))\right) K -2 \,.
\end{equation}
The trace $\Tr\left(\Pi(\rho)(\Pi(\rho)-\Pi(\rho^\beta \tau \rho^\beta))\right)$ is different from zero only if $ \supp\!\left(\rho^{\beta}\sigma \rho^{\beta}\right) \subset \supp(\rho)$. In this case, the derivative~\eqref{lower bound derivative} goes to infinity as $x^{-\alpha}$ as $x$ vanishes. This means that in the neighborhood of $x=0$ the function $Q_{\alpha,z}(\rho \| (1-x)\tau+x\sigma)$ is a strictly increasing (and continuous) function of $x$ and hence we can always find a sufficiently small $x$ such that $Q_{\alpha,z}(\rho \| (1-x)\tau+x\sigma) > Q_{\alpha,z}(\rho \| \tau)$. This implies that if $\tau \in \argmin_{\sigma \in \mathcal{F}} D_{\alpha,z}(\rho \| \sigma)$ then $ \supp\!\left(\rho^{\beta}\tau \rho^{\beta}\right) = \supp(\rho)$.  

We finally use Lemma~\ref{support} to show that  $ \supp\!\left(\rho^{\beta}\tau \rho^{\beta}\right) = \supp(\rho)$ if and only if $ \supp\!\left(\Pi(\rho)\tau \Pi(\rho)\right) =\supp(\rho)$. Indeed, we have that $\gamma \Pi(\rho) \leq \rho^{\beta} \tau  \rho^{\beta} \leq \beta \Pi(\rho)$ implies
\begin{align}
\lambda_{\min}\left( \rho^{-2\beta}\right) \Pi(\rho) \leq \gamma \rho^{-2\beta} \leq \Pi(\rho) \tau  \Pi(\rho) \leq \beta \rho^{-2\beta} \leq \beta \lambda_{\max}\left(\rho^{-2\beta}\right)  \Pi(\rho) \,,
\end{align}
where we used that since the negative power is taken only outside the kernel of $\rho$, it holds $\rho^{-\beta} \rho^{\beta} = \Pi(\rho)$.
Conversely, the relation $\gamma \Pi(\rho) \leq \Pi(\rho) \tau  \Pi(\rho) \leq \beta \Pi(\rho) $ implies that 
\begin{align}
\gamma \lambda_{\min}\left( \rho^{2\beta}\right) \Pi(\rho) \leq \gamma \rho^{2\beta}  \leq  \rho^{\beta} \tau  \rho^{\beta} \leq \beta \rho^{2\beta} \leq \beta \lambda_{\max}\left(\rho^{2\beta}\right) \Pi(\rho) \,.
\end{align}
This means that the condition $ \supp\!\left(\Pi(\rho)\tau \Pi(\rho)\right) =\supp(\rho)$ is a necessary condition for $\tau$ to be a minimum of $\mathfrak{D}_{\alpha,z}(\rho)$. 
\end{proof}

\begin{proof}[Proof for $|(1-\alpha)/z| \neq 1 \, \wedge \, \alpha \in (0,1)$]
The strategy for the proof in this range is similar to the one provided for the case $z=1-\alpha$. We first calculate the Fr\'echet derivative along the direction $\sigma$ with full support 
\begin{align}
&\partial_x Q_{\alpha,z}(\rho \| (1-x)\tau+x\sigma)= z \Tr\left( \chi_{\alpha,z}(\rho,(1-x)\tau+x\sigma) \partial_x ((1-x)\tau + x \sigma)^\frac{1-\alpha}{z}\right) \,.
\end{align}
We now use Lemma~\ref{technical} to get
\begin{align}
\label{sum derivative}
&\partial_x Q_{\alpha,z}(\rho \| (1-x)\tau+x\sigma) \notag \\
& \qquad \qquad \qquad
=(1-\alpha)K_{\alpha,z} \Tr \left(\chi_{\alpha,z}(\rho,(1-x)\tau+x\sigma)\int_0^\infty \frac{\sigma-\tau}{((1-x)\tau + x \sigma+t)^{2}} t^\frac{1-\alpha}{z} \d t \right) \,,
\end{align}
where we defined the positive constant $K_{\alpha,z}:=\sinc\left(\pi\frac{{1-\alpha}}{z}\right)$. 
We first bound the first term in the sum on the r.h.s of equation~\eqref{sum derivative}. We have
\begin{align}
& \Tr \left(\chi_{\alpha,z}(\rho,(1-x)\tau+x\sigma) \int_0^\infty \frac{\sigma}{((1-x)\tau + x \sigma+t)^{2}} t^\beta \d t \right) \\
&\qquad \qquad \qquad \qquad \geq \lambda_{\min}(\sigma)K^{-1}_{\alpha,z}  \Tr \left( \chi_{\alpha,z}(\rho,(1-x)\tau+x\sigma) ((1-x)\tau+x\sigma)^{\frac{1-\alpha}{z}-1}\right) \,,
\end{align}
where we used that for a positive operator $A>0$ it holds
\begin{equation}
\int_0^\infty \frac{1}{(A+ t)^2} t^{\beta} \d t =  \frac{A^{\beta-1} }{\sinc{(\pi \beta)}}  \; \quad \text{for}\;\;   \beta \in (-1,0) \cup (0,1) \,.
\end{equation}
Let us denote $\beta:=\alpha/(2z)$. We distinguish two cases. In the following, we repeatedly use that for $A,B \geq 0$ such that $[A,B] = 0$ it holds $A \leq B$ if and only if $A^\gamma \leq B^\gamma$ for any $\gamma > 0$. Moreover, if $\gamma <0$, the latter inequality goes in the opposite direction.

For $z \leq 1$ we have that $((1-x)\tau + x \sigma) \leq \Pi(\tau) + x \Pi^c(\tau)$ and $\rho^\beta((1-x)\tau + x \sigma)^\frac{1-\alpha}{z}\rho^\beta \leq \Pi(\rho)$ where we denoted the projector $\Pi^c(\tau):=\mathds{1}-\Pi(\tau)$. This implies that 
\begin{align}
\Tr \left( \chi_{\alpha,z}(\rho,(1-x)\tau+x\sigma) ((1-x)\tau+x\sigma)^{\frac{1-\alpha}{z}-1}\right) 
\label{op-ant}
&\geq \Tr \left( \rho^{2\beta} (\Pi(\tau) + x^{\frac{1-\alpha}{z}-1} \Pi^c(\tau)) \right) \\
\label{first functional}
& \geq  x^{\frac{1-\alpha}{z}-1}  \Tr(\Pi(\rho)\Pi^c(\tau)) K_1 \,,
\end{align}
where we denoted and $K_1:= \lambda^{2\beta}_{\min}(\rho)$ is a positive constant. In~\eqref{op-ant} we used that for $|(1-\alpha)/z| \neq \, 1 \wedge \,\alpha \in (0,1)$ we have $(1-\alpha)/z-1 \in (-1,0)$.

For $z>1$ we use that $\rho^\beta((1-x)\tau + x \sigma)^\frac{1-\alpha}{z}\rho^\beta \geq x^\frac{1-\alpha}{z}\lambda_{\min}(\sigma)\lambda^{2\beta}_{\min}(\rho) \Pi(\rho)$. This implies that
\begin{align}
&\Tr \left( \chi_{\alpha,z}(\rho,(1-x)\tau+x\sigma) ((1-x)\tau+x\sigma)^{\frac{1-\alpha}{z}-1}\right)  \\
&\qquad \qquad \qquad \geq \left(x^\frac{1-\alpha}{z}\lambda_{\min}(\sigma)\lambda^{2\beta}_{\min}(\rho)\right)^{z-1} \Tr \left( \rho^{2\beta} (\Pi(\tau) + x^{\frac{1-\alpha}{z}-1} \Pi^c(\tau)) \right) \\
\label{second functional}
&  \qquad \qquad \qquad \geq x^{-\alpha} \Tr(\Pi(\rho)\Pi^c(\tau)) K_2 \,,
\end{align}
where $K_2:= \lambda^{z-1}_{\min}(\sigma)\lambda^{2\beta z}_{\min}(\rho)$ is a positive constant.

We now bound the second term in the sum on the r.h.s of equation~\eqref{sum derivative}. We use that $\tau \leq \tau + (x/(1-x))\sigma$ to obtain
\begin{align}
\Tr \left(\chi_{\alpha,z}(\rho,(1-x)\tau+x\sigma) \int_0^\infty \frac{\tau}{((1-x)\tau + x \sigma+t)^{2}} t^\beta \d t \right) &\leq K^{-1}_{\alpha,z} \frac{1}{1-x} Q_{\alpha,z}(\rho \| (1-x)\tau + x \sigma) \\
&  \leq  2 K^{-1}_{\alpha,z} \,,
\end{align}
where we assumed $x \leq 1/2$. Combining the above results, we get
\begin{equation}
\partial_x Q_{\alpha,z}(\rho \| (1-x)\tau+x\sigma) \geq
\begin{cases}
x^{\frac{1-\alpha}{z}-1}  \Tr(\Pi(\rho)\Pi^c(\tau)) C_1 - 2 D & z \leq 1 \\
x^{-\alpha} \Tr(\Pi(\rho)\Pi^c(\tau)) C_2 - 2 D& z>1
\end{cases} \,,
\end{equation}
where $C_i:= K_i \lambda_{\min}(\sigma)K^{-1}_{\alpha,z}$, and $D:= 1-\alpha$ are two positive constants. 
The trace $\Tr(\Pi(\rho)\Pi^c(\tau))$ is different from zero if $ \supp\!\left(\tau\right) \subset \supp(\rho)$. In this case, the above derivative diverges as $x$ vanishes for both $z \leq 1$ and $z>1$. As we discussed above for the case $z=\alpha-1$, this means that in the neighborhood of $x=0$ the function $Q_{\alpha,z}(\rho \| (1-x)\tau+x\sigma)$ is a strictly increasing (and continuous) function of $x$ and hence we could always find a sufficiently small $x$ such that $Q_{\alpha,z}(\rho \| (1-x)\tau+x\sigma) > Q_{\alpha,z}(\rho \| \tau)$. This means that, as for the range $\alpha \geq 1$, the condition $ \supp(\rho) \subseteq \supp(\tau)$ is a necessary condition for any optimizer in the range $|(1-\alpha)/z| \neq 1 \wedge \alpha \in (0,1)$.
\end{proof}

To formulate the theorem below in a compact way, we introduce the following positive operator. Let $\rho,\tau \in \qstate$. We define
\begin{align}
\label{equation problem}
&\Xi_{\alpha,z}(\rho,\tau) :=  \begin{dcases}
\chi_{\alpha,1-\alpha}(\rho,\tau) &  \textup{if} \; z=1-\alpha \\
\tau^{-1}  \chi_{\alpha,\alpha-1}(\rho, \tau) \tau^{-1}  &  \textup{if} \; z=\alpha-1 \\
K_{\alpha,z} \int_0^\infty \frac{\chi_{\alpha,z}(\rho,\tau)}{(\tau + t)^2} t^\frac{1-\alpha}{z} \d t & \textup{if} \; |(1-\alpha)/z| \neq 1
\end{dcases} \,,
\end{align}
where $K_{\alpha,z}:=\sinc\left(\pi\frac{{1-\alpha}}{z}\right)$ is a positive constant and we defined the positive operator $\chi_{\alpha,z}(\rho,\tau):=  \rho^\frac{\alpha}{2z}( \rho^{\frac{\alpha}{2z}}\tau^{\frac{1-\alpha}{z}} \rho^{\frac{\alpha}{2z}})^{z-1} \rho^\frac{\alpha}{2z}$.
Note that the two lines $z=1-\alpha$ and $z=\alpha-1$ are located at the boundaries of the DPI region for $\alpha \leq 1/2$ and $\alpha \geq 2$, respectively (see Fig.~\ref{fig: alpha-z}). We are now ready to state the necessary and sufficient conditions for the optimizer(s) of~\eqref{problem F}.

\begin{theorem}
\label{main Theorem}
Let $\rho$ be a quantum state and $(\alpha,z) \in \mathcal{D}$. Then $\tau \in \argmin_{\sigma \in \mathcal{F}} D_{\alpha,z}(\rho \| \sigma)$ if and only if $\tau \in S_{\alpha,z}(\rho)$ and $\textup{Tr}(\sigma\, \Xi_{\alpha,z}(\rho,\tau)) \leq  Q_{\alpha,z}(\rho \|\tau)$ for all $\sigma \in \mathcal{F}$.
\end{theorem}
We remark the above optimization is much easier to perform than the original optimization in~\eqref{problem F} since, by the linearity of the trace, it can be restricted to the set of pure free states. The reader may argue that the problem simplifies only once an ansatz to the problem is provided and that in general finding such an ansatz is not an easy task. However, to investigate the additivity problem, the ansatz is given `for free', being the tensor product of the optimizers of the marginal problems (see the discussion below).  

\begin{proof}
Since the function $Q_{\alpha,z}$ is jointly concave/convex in the range $(\alpha,z) \in \mathcal{D}$ where $D_{\alpha,z}$ satisfies the data-processing inequality, as we discussed above, it follows that the condition of local minimum is both sufficient and necessary. 

The condition of the theorem $\textup{Tr}(\sigma\, \Xi_{\alpha,z}(\rho,\tau)) \leq  Q_{\alpha,z}(\rho \|\tau)$ for all $\sigma \in \mathcal{F}$ follows by the condition of local minimum for $\tau$, i.e. $\partial_\sigma Q_{\alpha,z}(\rho \| \tau) \leq 0$ and $\partial_\sigma Q_{\alpha,z}(\rho \| \tau) \geq 0$ for any $\sigma \in \mathcal{F}$ for $\alpha<1$ and $\alpha >1$, respectively.

Indeed, let $\tau \in S_{\alpha,z}(\rho)$. For $z = 1-\alpha$ ($\alpha \in (0,1)$), we have
\begin{align}
\partial_{\sigma} Q_{\alpha,1-\alpha}(\rho \|\tau) = z \Tr(\chi_{\alpha,1-\alpha}(\rho,\tau) (\sigma-\tau)) = z\Tr(\chi_{\alpha,1-\alpha}(\rho,\tau)\sigma) -  z Q_{\alpha,1-\alpha}(\rho \| \tau) \,.
\end{align}

For $z=\alpha-1$ ($\alpha > 1$) we use that $\partial \sigma(x)^{-1}/\partial x = - \sigma(x)^{-1} (\partial \sigma(x)/\partial x)  \sigma(x)^{-1}$ to obtain
\begin{align}
\partial_\sigma Q_{\alpha,\alpha-1}(\rho \| \tau) &= - z \Tr\left( \chi_{\alpha,\alpha-1}(\rho, \tau) \tau^{-1}(\sigma - \tau) \tau^{-1}  \right)\\
& = - z\Tr\left(\tau^{-1}  \chi_{\alpha,\alpha-1}(\rho,\tau) \tau^{-1}\sigma\right) + zQ_{\alpha,\alpha-1}(\rho \| \tau) \,.
\end{align}

For $(1-\alpha)/z \in (-1,0) \cup (0,1)$ we use Lemma~\ref{eq lemma} to obtain
\begin{align}
\partial_\sigma Q_{\alpha,z}(\rho \| \tau) =(1-\alpha)K_{\alpha,z}  \Tr \left(\sigma \int_0^\infty \frac{\chi_{\alpha,z}(\rho,\tau) }{(\tau + t)^{2}} t^\frac{1-\alpha}{z} \d t \right) - (1-\alpha) Q_{\alpha,z}(\rho\|\tau) \,.
\end{align}
Note that in the range $(\alpha,z) \in \mathcal{D} \wedge |(1-\alpha)/z| \neq 0,1$ , the condition $(1-\alpha)/z \in (-1,0)\cup (0,1)$ is always satisfied. The theorem for $|(1-\alpha)/z| \neq 0,1$ then follows by combining the above results with the local minimum conditions above-mentioned.  For the Umegaki relative entropy, i.e. $\alpha=1$ (or $|(1-\alpha)/z| = 0$), the result has already been obtained in~\cite{vedral1998entanglement} and can be recovered by taking the limit $\alpha \rightarrow 1$ of the result for the range $|1-\alpha)/z| \neq 1$ (see the discussion below).
\end{proof}

As we show below, many states we consider commute with their optimizer. In this case, the above theorem considerably simplifies. 

\begin{corollary}
\label{commuting}
Let $\rho$ be a quantum state and $(\alpha,z) \in \mathcal{D}$. Then, a state $\tau$ satisfying $[\rho,\tau]= 0$ belongs to the set $\tau \in \argmin_{\sigma \in \mathcal{F}} D_{\alpha,z}(\rho \| \sigma)$ if and only if $\supp(\rho)\subseteq \supp(\tau)$ and $\textup{Tr}(\sigma \Xi_{\alpha}(\rho,\tau)) \leq  Q_\alpha(\rho \|\tau)$ for all $\sigma \in \mathcal{F}$ where $\Xi_{\alpha}(\rho,\tau) = \rho^\alpha \tau^{-\alpha}$ and $Q_\alpha(\rho \|\tau) = \textup{Tr}(\rho^\alpha \tau^{1-\alpha})$. 
\end{corollary}

\begin{proof}
If $[\rho,\tau]= 0$, we use that for a positive operator $A>0$
\begin{equation}
 \int_0^\infty \frac{1}{(A+ t)^2} t^{\beta} \d t =  \frac{A^{\beta-1} }{\sinc{(\pi \beta)}} \quad \text{for} \quad \beta \in (-1,0) \cup (0,1) \,,
\end{equation}
to further simplify the expression~\eqref{equation problem}. We obtain for $(1-\alpha)/z \in (-1, 1)$
\begin{equation}
\textup{Tr}(\sigma \Xi_{\alpha,z}(\rho,\tau)) =   \textup{Tr}(\sigma {\Xi}_{\alpha}(\rho,\tau)) \,,
\end{equation}
where $\Xi_{\alpha}(\rho,\tau) = \rho^\alpha \tau^{-\alpha}$. The cases $|(1-\alpha)/z| = 1$ follow easily by direct computation of~\eqref{equation problem} under the assumption $[\rho,\tau]= 0$. Finally, the support condition  $\supp(\rho) = \supp\!\left(\Pi(\rho)\tau \Pi(\rho)\right)$ which holds for $z=1-\alpha$ reduces to $\supp(\rho)\subseteq \supp(\tau)$ for $[\rho, \tau]=0$.
\end{proof}

\begin{remark}
The necessary and sufficient conditions for the Umegaki relative entropy case ($\alpha=1$) have already been pointed out in~\cite{vedral1998entanglement}. In this case, $\tau$ is the solution of the minimization problem~\eqref{problem F} if and only if $\textup{Tr}(\sigma\, \Xi_{1,1}(\rho,\tau)) \leq 1$ for all $\sigma \in \mathcal{F}$ where 
\begin{align}
\label{relative entropy}
\Xi_{1,1}(\rho,\tau) =  \int_0^\infty \frac{\rho}{(\tau+t)^2} \d t \,.
\end{align}
Note that we can recover this result by taking the limit $\alpha \rightarrow 1$ of the result for the range $|1-\alpha)/z| \neq 1$ in Theorem~\ref{main Theorem}. Therefore, we have that $\lim_{\alpha \rightarrow 1} \partial_\sigma D_{\alpha,z}(\rho \| \tau) = \partial_\sigma \lim_{\alpha \rightarrow 1} D_{\alpha,z}(\rho \| \tau)$ for all $z>0$, i.e. we can interchange the limit with the Frech\'et derivative. Moreover, if $[\rho,\tau]=0$, we get $\Xi_{1,1}(\rho,\tau)=\rho \tau^{-1}$. We again recover the latter result by taking the limit $\alpha \rightarrow 1$ in Corollary~\ref{commuting}.

Another interesting limit is the limit $\alpha \rightarrow 0$ with $z=1$. In this case, we obtain the min-relative entropy (see equation~\eqref{Dmin}). Using the explicit form of the min-relative entropy, it is easy to see that $\Xi_{0,1}(\rho) = \Pi(\rho)$. Note that we can recover this result by taking the limit $\alpha \rightarrow 0$ with $z=1$ of the result for the range $|(1-\alpha)/z| = 1$ in Theorem~\ref{main Theorem}.
\end{remark}

\section{Additivity of the $\alpha$-$z$ R\'enyi relative entropies of entanglement}
\label{Additivity}
In this section, we address our main results.  We prove that the monotones  $\mathfrak{D}_{\alpha,z}$ are additive whenever one state belongs to a special class. In the next section we show below that bipartite pure, maximally correlated, GHZ, Bell diagonal, generalized Dicke, isotropic, and separable states are some examples of these classes. To the best of our knowledge, such additivity results have previously only been established for the case where \textit{both} states belong to some class of states or only for some value of the parameters $\alpha$ and $z$. Some examples include bipartite pure states~\cite{vedral1998entanglement}, maximally correlated states~\cite{zhu2017coherence}, Bell diagonal~\cite{zhu2010additivity} or generalized Dicke states~\cite{zhu2010additivity}.

Before focusing on entanglement theory, we first show how the conditions that we derived in Section~\ref{condition} provide a powerful tool to investigate additivity questions in any convex resource theory by considering the resource theory of coherence. For this resource theory, the monotone based on the $\alpha$-$z$ R\'enyi relative entropy is additive for any state. This result follows immediately for the case $\alpha=z=1$ from the closed-form expression of the relative entropy of coherence~\cite[Theorem 6]{Winter} and has been extended for all $z=\alpha$ and $z=1$ in~\cite[Theorem 3]{zhu2017coherence}. We generalize this result to all $(\alpha,z) \in \mathcal{D}$ by providing a simple and independent proof that hinges on the necessary and sufficient conditions derived in the previous section.

\subsection{Additivity of the $\alpha$-$z$ R\'enyi relative entropies of coherence}
\label{Additivity coherence}
Coherence is defined with respect to a particular basis dictated by the physical problem under consideration~\cite{Winter}. If $\{|i\rangle, i=1, ...,d\} $ is such a basis, a state is called free if it is diagonal in this basis, namely if it is of the form $ \sum p_i |i\rangle \! \langle i |$ with $\sum p_i=1$. We call these states \textit{incoherent states}, and we denote this set (free set) by $\mathcal{I}$. States that are not \textit{incoherent states} are resourceful, and we refer to them as \textit{coherent states}. In this specific resource theory, the monotones based on the $\alpha$-$z$ R\'enyi relative entropy are additive. In the following, we state the result for the points $(\alpha,z) \in \mathcal{D}$. The result for the limiting cases follows by taking the corresponding limits (see Appendix~\ref{lower-semicontinuous} for a more detailed discussion).
\begin{theorem}
\label{coherence}
Let $\rho_1,\rho_2 \in \qstate$ be two states and $(\alpha,z) \in \mathcal{D}$. Then we have
\begin{equation}
\mathfrak{D}^\mathcal{I}_{\alpha,z}(\rho_1 \otimes \rho_2) = \mathfrak{D}^\mathcal{I}_{\alpha,z}(\rho_1) + \mathfrak{D}^\mathcal{I}_{\alpha,z}(\rho_2) \,.
\end{equation}
\end{theorem}

\begin{proof}
To prove additivity, we show that $\tau_1 \otimes \tau_2 \in \argmin_{\sigma \in \mathcal{F}} D_{\alpha,z}(\rho_1 \otimes \rho_2 \| \sigma)$ if $\tau_1$ and $\tau_2$ are some optimizers of the marginal problems, i.e. $\tau_1 \in \argmin_{\sigma \in \mathcal{F}} D_{\alpha,z}(\rho_1 \| \sigma)$ and $\tau_2 \in \argmin_{\sigma \in \mathcal{F}} D_{\alpha,z}(\rho_2 \| \sigma)$. Since the $\alpha$-$z$ R\'enyi relative entropy is additive~\cite{audenaert2015alpha}, the latter condition implies the additivity of the related monotones. Indeed, we have  $\mathfrak{D}^{\mathcal{I}}_{\alpha,z}(\rho_1 \otimes \rho_2) = D_{\alpha,z}(\rho_1 \otimes \rho_2 \| \tau_1 \otimes \tau_2) = D_{\alpha,z}(\rho_1 \|\tau_1) + D_{\alpha,z}(\rho_2 \|\tau_2) = \mathfrak{D}^\mathcal{I}_{\alpha,z}(\rho_1) + \mathfrak{D}^\mathcal{I}_{\alpha,z}(\rho_2)$. 
Therefore, to prove additivity, according to Theorem~\ref{main Theorem}, we need to show that 
\begin{equation}
\label{coherence functional}
\Tr(\sigma \Xi_{\alpha,z}(\rho_1\otimes \rho_2, \tau_1 \otimes \tau_2)) \leq Q_{\alpha,z}(\rho_1 \otimes \rho_2 \| \tau_1 \otimes \tau_2)= Q_{\alpha,z}(\rho_1 \| \tau_1) Q_{\alpha,z}(\rho_2 \| \tau_2) \; \forall \sigma \in \mathcal{I} \,.
\end{equation}
In the last equality, we used that $Q_{\alpha,z}$ is multiplicative. Note that we always have $\tau_1 \otimes \tau_2 \in S_{\alpha,z}(\rho_1 \otimes \rho_2)$ since, by assumption, the states $\tau_1$ and $\tau_2$ are optimizers of the marginal problems and hence they satisfy $\tau_1 \in S_{\alpha,z}(\rho_1)$ and $\tau_2 \in S_{\alpha,z}(\rho_2)$. 

By linearity of the trace, we can restrict $\sigma$ above to be a pure incoherent state. We can write any pure incoherent state as $\sigma= |i,j\rangle$. In the following, we use the decomposition for the optimizers $\tau_i=\sum t_{i,j}|j\rangle \! \langle j|$. We have for $|(1-\alpha)|/z \neq 1$
\begin{align}
\langle i,j| \Xi_{\alpha,z}&(\rho_1\otimes \rho_2, \tau_1 \otimes \tau_2)|i,j \rangle \\
& = K_{\alpha,z}\int_0^\infty \left(t_{1,i} t_{2,j} +k\right)^{-1} \left(t_{1,i} t_{2,j} +k\right)^{-1} k^\frac{1-\alpha}{z}\d k  \langle i | \chi_{\alpha,z}(\rho_1,\tau_1) |i \rangle \! \langle j | \chi_\alpha(\rho_2,\tau_2) |j \rangle \notag \\
&= \big(t_{1,i} t_{2,j}\big)^{\frac{1-\alpha}{z}-1} \langle i | \chi_{\alpha,z}(\rho_1,\tau_1) |i \rangle \! \langle j | \chi_{\alpha,z}(\rho_2,\tau_2) |j \rangle \leq Q_{\alpha,z}(\rho_1\|\tau_1) Q_{\alpha,z}(\rho_2\|\tau_2) \,,
\end{align}
where we used that, if $\tau = \sum t_{j}|j\rangle \! \langle j|$ is an optimizer of $\rho$, the condition $ t_l^{\frac{1-\alpha}{z}-1} \langle l | \chi_{\alpha,z}(\rho,\tau) |l \rangle \leq Q_{\alpha,z}(\rho \|\tau)$ holds for any $l$. Indeed, if we choose the pure  incoherent state $\sigma=|l\rangle$, we obtain
\begin{align}
& \Tr(\sigma \Xi_{\alpha,z}(\rho, \tau)) = \langle l|  \Xi_{\alpha,z}(\rho, \tau) | l\rangle = t_l^{\frac{1-\alpha}{z}-1} \langle l | \chi_{\alpha,z}(\rho,\tau) |l \rangle \,.
\end{align}
Since $\tau$ is an optimizer of $\rho$ by assumption, the above equality together with Theorem~\ref{main Theorem} implies that $ t_l^{\frac{1-\alpha}{z}-1} \langle l | \chi_{\alpha,z}(\rho,\tau) |l \rangle \leq Q_{\alpha,z}(\rho \|\tau)$ for any $l$. Hence, the inequality~\eqref{coherence functional} is satisfied and additivity follows. The result for the case $|(1-\alpha)/z|=1$ can be obtained by taking the limits $z \rightarrow 1-\alpha$ and $z \rightarrow \alpha -1$ (see Appendix~\ref{lower-semicontinuous} for a more detailed discussion). 
\end{proof}

\subsection{Additivity of the $\alpha$-$z$ R\'enyi entropy of entanglement: Commuting case}
\label{First class}
We now prove that in entanglement theory, the monotones $\mathfrak{D}_{\alpha,z}$ are additive when one state commutes with its optimizer and the alpha power of its product with the inverse of its optimizer has positive entries in a product basis. We remark that the latter condition is very important since, as we show in Section~\ref{Counterexample}, there exist states that commute with their optimizer but do not have positive entries in a product basis for which the $\mathfrak{D}_{\alpha,z}$ are not additive.  Separable states clearly satisfy these conditions, since the optimizer coincides with the state itself. This is not surprising since for separable states, the proof already follows from the fact that the $\mathfrak{D}_{\alpha,z}$ are subadditive, equal zero if and only if the state is separable, and they are non-increasing under partial trace.  More interestingly, as we show later, the Bell diagonal, generalized Dicke, and the isotropic states belong to this class. As before, we state the result for the points $(\alpha,z) \in \mathcal{D}$. The result for the limiting cases follows by taking the corresponding limits (see Appendix~\ref{lower-semicontinuous} for a more detailed discussion).

In~\cite{rains1999bound} (see also~\cite[Lemma 1]{audenaert2002asymptotic}) the author showed that if the optimizer $\tau$ of $\rho$ satisfies $[\rho,\tau]=0$ and $|(\rho \sigma^{-1})^{T_2}| \leq \mathds{1}$ where $T_2$ denotes the partial transpose on the second system, then the relative entropy optimized over the set of PPT states is weakly additive. Moreover, if the state satisfies the stronger condition $0 \leq (\rho \sigma^{-1})^{T_2} \leq \mathds{1}$, then the REEP is strongly additive for $\rho$, i.e. it is additive for the tensor product of $\rho$ with any other state. In the following, we establish a similar result for the relative entropy of entanglement. We remark that, although the condition $|(\rho \sigma^{-1})^{T_2}| \leq \mathds{1}$ holds for several states that commute with their optimizer (e.g. for all two-qubit states~\cite{miranowicz2008closed} or for orthogonally invariant states~\cite{audenaert2002asymptotic}), the condition for strong additivity $0 \leq (\rho \sigma^{-1})^{T_2}$ does not typically hold due to the fact that the state $\rho \sigma^{-1}$ is not PPT. However, our result for strong additivity of the relative entropy of entanglement applies to several states that do not satisfy the condition just mentioned. An explicit example is provided by the isotropic states (see Section~\ref{Analytics}). Finally, we note that the result for the geometric measure of entanglement $\alpha=z=1/2$ was not known for any specific states even for the case where both states (or, in general $n$ states) are equal.

\begin{theorem}
\label{CC}
Let $\rho_1$ be a N-partite state and $(\alpha,z) \in \mathcal{D}$ . Moreover, let  $\tau_1 \in \argmin_{\sigma \in \textup{SEP}} D_{\alpha,z}(\rho_1 \| \sigma)$. If $[\rho_1,\tau_1]=0$ and $\rho_1^\alpha\tau_1^{-\alpha}$ has positive entries in a product basis, then for any N-partite state $\rho_2$ we have that
\begin{equation}
\mathfrak{D}_{\alpha,z}(\rho_1 \otimes \rho_2) = \mathfrak{D}_{\alpha,z}(\rho_1) + \mathfrak{D}_{\alpha,z}(\rho_2)\,.
\end{equation}
\end{theorem}
\begin{proof}
Similarly to the proof of Theorem~\ref{coherence}, to prove additivity we show that if $\tau_1$ is an optimizer of $\rho_1$ such that $[\rho_1,\tau_1]=0$, and $\rho_1^\alpha\tau_1^{-\alpha}$ has positive entries in a product basis, then an optimizer of $\rho_1 \otimes \rho_2$ is $\tau_1 \otimes \tau_2$ where $\tau_2$ is any optimizer of $\rho_2$. 
According to Theorem~\ref{main Theorem}, we want to show that 
\begin{equation}
\label{functional}
\Tr(\sigma \Xi_{\alpha,z}(\rho_1\otimes \rho_2, \tau_1 \otimes \tau_2)) \leq  Q_{\alpha}(\rho_1 \| \tau_1) Q_{\alpha,z}(\rho_2 \| \tau_2) \quad \forall \sigma \in \SEP\,,
\end{equation}
where we denoted with $\tau_1$ and $\tau_2$ some optimizers of $\rho_1$ and $\rho_2$, respectively. We use the spectral decomposition $\tau_1 = \sum_l t_{1,l}|\psi_l\rangle \! \langle \psi_l|$ and $\tau_2 = \sum_r t_{2,r}|\xi_r\rangle \! \langle \xi_r|$. Here, $|\psi_l\rangle$ is the common eigenbasis of $\rho_1$ and $\tau_1$.   In the range $|(1-\alpha)/z|\neq 1$ we have
\begin{align}
 & \Xi_{\alpha,z}(\rho_1 \otimes \rho_2, \tau_1 \otimes \tau_2) \\
&\qquad = \, K_{\alpha,z} \sum \limits_{l, r_1,r_2}  \int_0^\infty \left(t_{1,l} t_{2,r_1} +k\right)^{-1} \left(t_{1,l} t_{2,r_2} +k\right)^{-1} k^\frac{1-\alpha}{z}\d k \;   \;| \psi_{l} \rangle \! \langle \psi_{l}| \chi_{\alpha,z}(\rho_1,\tau_1) |\psi_{l} \rangle \! \langle \psi_{l}| \;  \notag\\
& \qquad \quad  \otimes | \xi_{r_1} \rangle \! \langle \xi_{r_1}| \chi_{\alpha,z}(\rho_2,\tau_2) |\xi_{r_2} \rangle \! \langle \xi_{r_2}|  \\
& \qquad = \, K_{\alpha,z} \sum \limits_{l} t_{1,l}^{\frac{1-\alpha}{z}-1} | \psi_{l} \rangle \! \langle \psi_{l}| \chi_{\alpha,z}(\rho_1,\tau_1) |\psi_{l} \rangle \! \langle \psi_{l}| \;  \notag \\
& \qquad \quad \otimes  \sum_{r_1,r_2} \int_0^\infty \left( t_{2,r_1} +k\right)^{-1} \left(t_{2,r_2} +k\right)^{-1} k^\frac{1-\alpha}{z}\d k |\xi_{r_1} \rangle \! \langle \xi_{r_1}| \chi_{\alpha,z}(\rho_2,\tau_2) |\xi_{r_2} \rangle \! \langle \xi_{r_2}| \\
\label{penultima}
&\qquad = \; \rho_1^\alpha\tau_1^{-\alpha} \otimes \Xi_{\alpha,z}( \rho_2, \tau_2)\\
& \qquad  =\;\Xi_{\alpha}( \rho_1, \tau_1) \otimes \Xi_{\alpha,z}( \rho_2, \tau_2) \,,
\end{align}
where in the second equality we changed the measure $k \rightarrow t_{1,l} k$. Moreover, in~\eqref{penultima} we used that $ \sum_{k} t_{1,l}^{\frac{1-\alpha}{z}-1} | \psi_{l} \rangle \! \langle \psi_{l}| \chi_{\alpha,z}(\rho_1,\tau_1) |\psi_{l} \rangle \! \langle \psi_{l}| = \rho_1^\alpha\tau_1^{-\alpha}$. In the last equality, we used the definition $\Xi_{\alpha}( \rho_1, \tau_1):=\rho_1^\alpha\tau_1^{-\alpha}$ introduced in Corollary~\ref{commuting}.
We then have, for any $\sigma \in \SEP$ and $(\alpha,z) \in \mathcal{D}$,
\begin{align}
 \Tr(\sigma \Xi_{\alpha,z}(\rho_1 \otimes \rho_2, \tau_1 \otimes \tau_2)) & \leq  \max \limits_{\sigma \in \SEP}\Tr(\sigma \Xi_{\alpha,z}(\rho_1\otimes \rho_2, \tau_1 \otimes \tau_2))\\
 &  =  \max \limits_{\sigma \in \SEP}\Tr(\sigma\Xi_{\alpha}( \rho_1, \tau_1) \otimes \Xi_{\alpha,z}( \rho_2, \tau_2)) \\
&  =  C_{1,\alpha} C_{2,\alpha,z} \max \limits_{\sigma \in \SEP}\text{Tr}(\sigma\hat{\Xi}_{\alpha}( \rho_1, \tau_1) \otimes \hat{\Xi}_{\alpha,z}( \rho_2, \tau_2)) \\
 \label{additivity}
& =  C_{1,\alpha} C_{2,\alpha,z} \max \limits_{\sigma \in \SEP}\text{Tr}(\sigma\hat{\Xi}_{\alpha}( \rho_1, \tau_1)) \max \limits_{\sigma \in \SEP}\text{Tr}(\sigma\hat{\Xi}_{\alpha,z}( \rho_2, \tau_2)) \\
& =  \max \limits_{\sigma \in \SEP}\Tr(\sigma \Xi_{\alpha}( \rho_1, \tau_1)) \max \limits_{\sigma \in \SEP}\Tr(\sigma \Xi_{\alpha,z}( \rho_2, \tau_2)) \\
& = \; Q_{\alpha}(\rho_1\| \tau_1) Q_{\alpha,z}(\rho_2\| \tau_2) \,,
\end{align} 
where we defined the positive constants $C_{i,\alpha}:= \text{Tr} (\Xi_{\alpha}( \rho_i, \tau_i) )$ for $i=1,2$, and for a positive operator $A$ we defined the quantum state $\hat{A}:= A/\Tr(A)$.
In~\eqref{additivity} we used the multiplicativity result in~\cite[Theorem 5]{zhu2010additivity} together with the assumption that $\Xi_{\alpha}( \rho_1, \tau_1) = \rho_1^\alpha\tau_1^{-\alpha}$ has positive entries in a product basis. The last equality follows from the fact that $\tau_1$ and $\tau_2$ are, by assumptions, optimizers of the marginal problems. Indeed, from Theorem~\ref{main Theorem}, if $\tau \in \argmin_{\sigma \in \SEP} D_{\alpha,z}(\rho \| \sigma)$, we have that $\Tr(\sigma\Xi_{\alpha,z}( \rho, \tau)) \leq Q_{\alpha,z}(\rho\| \tau)$ for any $\sigma \in \SEP$. Moreover, from Corollary~\ref{commuting},  if $\tau \in \argmin_{\sigma \in \SEP} D_{\alpha,z}(\rho \| \sigma)$, we have that $\Tr(\sigma\Xi_{\alpha}( \rho, \tau)) \leq Q_{\alpha}(\rho\| \tau)$ for any $\sigma \in \SEP$ if $\rho$ and $\tau$ commute. It is easy to see that the latter inequalities are saturated for $\sigma = \tau$.
The latter chain of inequalities proves the inequality~\eqref{functional}. The result for the case $|(1-\alpha)/z|=1$ can be obtained by taking the limits $z \rightarrow 1-\alpha$ and $z \rightarrow \alpha -1$ (see Appendix~\ref{lower-semicontinuous} for a more detailed discussion)
\end{proof}

\subsection{Additivity of the $\alpha$-$z$ R\'enyi entropy of entanglement: Maximally correlated states}
\label{Second class}
In this section, we prove that the monotones $\mathfrak{D}_{\alpha,z}$  are additive when one state is a maximally correlated state. Moreover, we generalize the latter result to include the GHZ state. 

A \textit{maximally correlated state} is a bipartite state of the form~\cite{rains1999bound}
\begin{equation}
\label{form}
\rho = \sum_{jk}\rho_{jk}|j,j\rangle \! \langle k,k| \,.
\end{equation}
Since pure bipartite states are maximally correlated states, this means that the monotone is also additive whenever one state is a pure bipartite state. 

We first derive necessary and sufficient conditions for an optimizer of a maximally correlated state. The following proposition is a fundamental ingredient to derive the additivity of the $\mathfrak{D}_{\alpha,z}$ when one state is maximally correlated. Moreover, in Section~\ref{Bipartite pure section} we show how it allows computing analytically the $\mathfrak{D}_{\alpha,z}$ for bipartite states. 

Let $\rho$ be a maximally correlated state. We denote with $|i,i\rangle$ is the basis such that $\rho$ has the form~\eqref{form}. We define $\mathcal{T}_{\rho}$ as the set of all separable states of the form $\sigma = \sum_i s_i |i,i\rangle \! \langle i,i|$. 

\begin{proposition}
\label{MC marginal}
Let $\rho$ be a maximally correlated state and $(\alpha,z) \in \mathcal{D}$. If $\tau= \sum_i t_{i} |i,i\rangle \! \langle i,i| \in \mathcal{T}_{\rho}$, then $\tau \in \argmin_{\sigma \in \textup{SEP}} D_{\alpha,z}(\rho \| \sigma)$ if and only if $\tau \in S_{\alpha,z}(\rho)$ and
\begin{equation}
\label{inequality}
\max_l t^{\frac{1-\alpha}{z}-1}_{l} \langle l,l|\chi_{\alpha,z}(\rho,\tau) |l,l\rangle \leq  Q_{\alpha,z}(\rho \|\tau) \, .
\end{equation}
\end{proposition}
\begin{proof}
We first prove that if $\tau =\sum_i t_{i} |i,i\rangle \! \langle i,i| \in \mathcal{T}_{\rho}$ satisfies the inequality~\eqref{inequality}, then $\tau$ satisfies the conditions of Theorem~\ref{main Theorem}, i.e. $\tau \in \argmin_{\sigma \in \textup{SEP}} D_{\alpha,z}(\rho \| \sigma)$. As we already discussed, since the trace function is linear, we can restrict the optimization problem of Theorem~\ref{main Theorem} to pure product states. Let us denote with $|i,j\rangle$ the basis such that $\rho$ has the form~\eqref{form}. We can expand any pure product state as $\sigma = \sum_{ij}a_i b_j |i,j\rangle$ with complex normalized coefficients $\sum_i |a_i|^2 =1$ and $\sum_i |b_i|^2 =1$. We obtain
\begin{align}
\label{MC single problem}
 \Tr(\sigma \Xi_{\alpha,z}(\rho, \tau)) 
 &\leq K_{\alpha,z} \sum_{i,j} |a_i||b_i||a_j||b_j| \left|  \int_0^\infty \langle i,i| \frac{\chi_{\alpha,z}(\rho ,\tau )}{(\tau + t)^{2}}|j,j\rangle t^\frac{1-\alpha}{z} \d t \right|\\
 &\leq K_{\alpha,z} \sum_{i,j} |a_i||b_i||a_j||b_j| \bigg(  \int_0^\infty \langle i,i|\frac{\chi_{\alpha,z}(\rho ,\tau )}{(\tau + t)^{2}}|i,i\rangle t^\frac{1-\alpha}{z} \d t  \\
& \hskip120pt \times   \int_0^\infty \langle j,j|\frac{\chi_{\alpha,z}(\rho ,\tau )}{(\tau + t)^{2}}|j,j\rangle t^\frac{1-\alpha}{z} \d t \bigg)^\frac{1}{2}\\
& \leq K_{\alpha,z} \max_k  \int_0^\infty \langle k,k|\frac{\chi_{\alpha,z}(\rho ,\tau )}{(\tau + t)^{2}}|k,k\rangle t^\frac{1-\alpha}{z} \d t  \sum_{i,j} |a_i||b_i||a_j||b_j| \\
& \leq \max_l t^{\frac{1-\alpha}{z}-1}_l \langle l,l|\chi_{\alpha,z}(\rho ,\tau ) |l,l\rangle \,,
\end{align}
In the first inequality we used that $\tau$ is of the form $ \tau = \sum t_i |i,i\rangle \! \langle i, i|$. In the second inequality we used that for a positive operator $A$ and vectors $\ket{v},\ket{w}$, the Cauchy-Schwarz inequality implies that $ |\langle v| A | w\rangle| =  |\langle v| \sqrt{A} \, \sqrt{A} | w\rangle| \leq \sqrt{\langle v| A | v\rangle \! \langle w| A | w\rangle}$. In the last inequality, we used the Cauchy-Schwarz inequality. The previous inequalities together with inequality~\eqref{inequality} and Theorem~\ref{main Theorem} implies $\tau \in \argmin_{\sigma \in \textup{SEP}} D_{\alpha,z}(\rho \| \sigma)$.

We now prove that any $\tau \in \mathcal{T}_{\rho}$ such that $\tau \in \argmin_{\sigma \in \textup{SEP}} D_{\alpha,z}(\rho \| \sigma)$ must satisfy the inequality~\eqref{inequality}. We choose the pure product state $\sigma = |l,l\rangle$. We have for any $l$
\begin{align}
 \Tr(\sigma \Xi_{\alpha,z}(\rho , \tau )) &= K_{\alpha,z} \int_0^\infty  \langle l,l|\frac{\chi_{\alpha,z}(\rho ,\tau )}{(\tau + t)^{2}}|l,l\rangle t^\frac{1-\alpha}{z} \d t \\
& = t^{\frac{1-\alpha}{z}-1}_{l} \langle l,l|\chi_{\alpha,z}(\rho ,\tau ) |l,l\rangle \,.
\end{align}
Since $\tau $ is an optimizer, the above equalities together with Theorem~\ref{main Theorem} imply that 
\begin{equation}
\label{marginal}
\max_l t^{\frac{1-\alpha}{z}-1}_{l} \langle l,l|\chi_{\alpha,z}(\rho ,\tau ) |l,l\rangle \leq  Q_{\alpha,z}(\rho  \|\tau ) \,.
\end{equation} 
\end{proof}
In Appendix~\ref{relations} we show that for any maximally correlated state $\rho$, there always exists an optimizer $\tau \in \argmin_{\sigma \in \textup{SEP}} D_{\alpha,z}(\rho \| \sigma)$ such that $\tau \in \mathcal{T}_{\rho}$.

\begin{theorem}
\label{TMC}
Let $\rho_1$ be a maximally correlated state and $(\alpha,z) \in \mathcal{D}$. Then, for any state $\rho_2$, we have
\begin{equation}
\mathfrak{D}_{\alpha,z}(\rho_1 \otimes \rho_2) = \mathfrak{D}_{\alpha,z}(\rho_1) + \mathfrak{D}_{\alpha,z}(\rho_2) \,.
\end{equation}
\end{theorem}

\begin{proof}
Similarly to the proof of Theorem~\ref{coherence}, according to Theorem~\ref{main Theorem}, to prove additivity it is enough to show that 
\begin{equation}
\label{functional MC}
\Tr(\sigma \Xi_{\alpha,z}(\rho_1\otimes \rho_2, \tau_1 \otimes \tau_2)) \leq  Q_{\alpha,z}(\rho_1 \| \tau_1) Q_{\alpha,z}(\rho_2 \| \tau_2) \qquad \forall \sigma \in \SEP \,,
\end{equation}
where $\tau_1$ and $\tau_2$ are optimizers of $\rho_1$ and $\rho_2$, respectively. 

We now prove that, if we choose an optimizer $\tau_1$ of $\rho_1$ in the set $ \mathcal{T}_{\rho_1}$, then the above inequality holds for any optimizer $\tau_2 $ of $\rho_2$. Since the trace function is linear, we can restrict the above optimization to pure product (separable) states. We can write any pure product state in the partition $AA'\!:\!BB'$ as $\sigma = \sum_{i}a_{i}|i , \phi_i\rangle_{AA'} \otimes \sum_{j} b_{j}|j , \psi_j  \rangle_{BB'} $ with positive normalized coefficients $a_i$ and $b_j$ and normalized $|\phi_i\rangle$ and $|\psi_i\rangle$. Here, we choose $|i,j\rangle_{AB}$ as the product basis such that $\rho_{\textup{MC}} = \sum_{jk}\rho_{jk}|j,j\rangle \! \langle k,k|_{AB} $.
We take the modulus and use the Cauchy-Schwarz inequality to get
\begin{align}
\label{MC}
 \Tr(\sigma \Xi_{\alpha,z}(\rho_1\otimes \rho_2, \tau_1 \otimes \tau_2)) &\leq \sum \limits_{i,j,i',j'} a_i b_j a_{i'} b_{j'}  |\langle i, j| \langle \phi_i ,\psi_j | \Xi_{\alpha,z}(\rho_1 \otimes \rho_2, \tau_1 \otimes \tau_2) |i' , j' \rangle |\phi_{i'} ,\psi_{j'}\rangle| \\
& \leq \sum \limits_{i,j,i',j'} a_i b_j a_{i'} b_{j'}  ( \langle i , j |\langle \phi_i ,\psi_j | \,\Xi_{\alpha,z}(\rho_1 \otimes \rho_2, \tau_1 \otimes \tau_2)\, |i , j  \rangle| \phi_{i} ,\psi_{j}\rangle \\
& \qquad \qquad  \times \langle i' , j' |\langle \phi_{i'} ,\psi_{j'} |\, \Xi_{\alpha,z}(\rho_1 \otimes \rho_2, \tau_1 \otimes \tau_2)\, |i' , j'  \rangle |\phi_{i'} ,\psi_{j'}\rangle)^\frac{1}{2} \,.
\end{align}
We now use the spectral decomposition $\tau_2 = \sum_{r}t_{2,r}|\xi_r \rangle \! \langle \xi_r|$. We have for $|(1-\alpha)/z| \neq 1$
\begin{align}
\label{hard}
& \langle i , j |\langle \phi_i ,\psi_j | \Xi_{\alpha,z}(\rho_1 \otimes \rho_2, \tau_1 \otimes \tau_2) |i , j  \rangle| \phi_{i} ,\psi_{j}\rangle \\
& \quad \quad = K_{\alpha,z} \delta_{ij}   \sum \limits_{r_1,r_2}  \int_0^\infty \left(t_{1,i}t_{2,r_1} +k\right)^{-1} \left(t_{1,i} t_{2,r_2} +k\right)^{-1} k^\frac{1-\alpha}{z}\d k  \\
& \qquad \qquad \qquad \qquad  \qquad \times \langle i,i| \chi_{\alpha,z}(\rho_1,\tau_1)|i,i\rangle  \langle  \phi_i ,\psi_i | \xi_{r_1} \rangle \! \langle \xi_{r_1}| \chi_{\alpha,z}(\rho_2,\tau_2) |\xi_{r_2} \rangle \! \langle \xi_{r_2}|\phi_{i} , \psi_{i}\rangle \notag \\
\label{equality}
& \quad \quad= K_{\alpha,z} \delta_{ij} t_{1,i}^{\frac{1-\alpha}{z}-1}  \sum \limits_{r_1,r_2}  \int_0^\infty \left(t_{2,r_1} +k\right)^{-1} \left(t_{2,r_2} +k\right)^{-1} k^\frac{1-\alpha}{z}\d k  \\
& \qquad \qquad \qquad \qquad  \qquad \times \langle i,i| \chi_{\alpha,z}(\rho_1,\tau_1)|i,i\rangle  \langle  \phi_i ,\psi_i | \xi_{r_1} \rangle \! \langle \xi_{r_1}| \chi_{\alpha,z}(\rho_2,\tau_2) |\xi_{r_2} \rangle \! \langle \xi_{r_2}|\phi_{i} , \psi_{i}\rangle \notag \\
& \quad \quad  =\,K_{\alpha,z} \delta_{i,j} t_{1,i}^{\frac{1-\alpha}{z}-1}  \langle i,i|\chi_{\alpha,z}(\rho_1,\tau_1) |i,i\rangle \!  \int_0^\infty \langle \phi_i \psi_i|(\tau_2 +k)^{-1} \chi_{\alpha,z}(\rho_2,\tau_2) (\tau_2 + k)^{-1}|\phi_i \psi_i\rangle k^\frac{1-\alpha}{z} \d k  \,  \\
& \quad \quad \leq  \,\delta_{i,j} t_{1,i}^{\frac{1-\alpha}{z}-1}  \langle i,i|\chi_{\alpha,z}(\rho_1,\tau_1) |i,i\rangle Q_{\alpha,z}(\rho_2\| \tau_2) \,.
\end{align}
In the first equality, we used the above-mentioned structure of the optimizer of $\tau_1$. In the second equality, we used the change of measure $k \rightarrow t_{1,i} k$. The last inequality follows from the fact that since by assumption $\tau_2$ is an optimizer of $\rho_2$, we have $ \Tr(\sigma \Xi_{\alpha,z}(\rho_2,\tau_2)) \leq  Q_{\alpha,z}(\rho_2\| \tau_2)$  for any $\sigma \in \SEP$. In addition, we used that any diagonal matrix elements of the positive operator  $\chi_{\alpha,z}(\rho_1,\tau_1)$ is positive. Therefore, for any $\sigma \in \SEP$, we have
\begin{align}
& \Tr(\sigma \Xi_{\alpha,z}(\rho_1\otimes \rho_2, \tau_1 \otimes \tau_2))\\
& \qquad \leq \sum \limits_{i,i'} a_i b_i a_{i'} b_{i'}   Q_{\alpha,z}(\rho_2\| \tau_2) \bigg((t_{1,i}t_{1,i'})^{\frac{1-\alpha}{z}-1} \langle i,i|\chi_{\alpha,z}(\rho_1,\tau_1) |i,i\rangle \! \langle i',i'|\chi_{\alpha,z}(\rho_1,\tau_1) |i',i'\rangle\bigg)^\frac{1}{2} \\
& \qquad \leq  Q_{\alpha,z}(\rho_1\| \tau_1) Q_{\alpha,z}(\rho_2\| \tau_2) \,,
\end{align}
where the last inequality follows from the Cauchy-Schwarz inequality and the marginal problem condition derived in Proposition~\ref{MC marginal}. 
Hence, the inequality~\eqref{functional MC} is satisfied and additivity holds. The result for the case $|(1-\alpha)/z|=1$ can be obtained by taking the limits $z \rightarrow 1-\alpha$ and $z \rightarrow \alpha -1$ (see Appendix~\ref{lower-semicontinuous} for a more detailed discussion)
\end{proof}

The above theorem can be easily generalized to include all the states of the form $|\rho\rangle = \sum_i \sqrt{p_i}|i,\dots,i\rangle$.  In particular, this class contains the GHZ state (see also Section~\ref{Bipartite pure section}).
The proof is similar, where now $\tau = \frac{1}{C}\sum_i p^{\beta}_i \ketbra{i, \dots,i}{i,\dots,i}$ (see Section~\ref{Bipartite pure section}). 

\begin{remark} 
Note that the previous results imply that whenever a state $\rho$ satisfies the conditions of Theorem~\ref{CC} or Theorem~\ref{TMC}, we have $\mathfrak{D}_{\alpha,z}(\rho)=\mathfrak{D}^\infty_{\alpha,z}(\rho)$, where the regularized or asymptotic monotone $\mathfrak{D}^\infty_{\alpha,z} $ is defined as $\mathfrak{D}^\infty_{\alpha,z}(\rho) = \lim_{n \rightarrow \infty}\frac{1}{n}\mathfrak{D}_{\alpha,z}(\rho^{\otimes n})$.
\end{remark}

\section{Counterexample to the additivity of the monotones based on a quantum relative entropy}
\label{Counterexample}
In this section, we prove that we prove that any monotone based on a quantum relative entropy is not additive for general states. In particular, additivity generally does not hold also when the two states are equal, which means that the latter monotones are also not weakly additive. To the best of our knowledge, this result is only known for $\mathfrak{D}_{\alpha,z}$ in the points $z=\alpha=1$ and $z=\alpha = \infty$. In particular, for these two points, in~\cite{vollbrecht2001entanglement} and~\cite{zhu2010additivity} the authors showed that the tensor product of two bipartite antisymmetric states provides a counterexample to the additivity of these monotones for general states. Moreover, in these two points, the value of the monotones for bipartite antisymmetric states and tensor product of bipartite antisymmetric does not depend on the value of $z$ and $\alpha$. Since the monotones $\mathfrak{D}_{\alpha,\alpha}$ are increasing functions of $\alpha$, the latter results already imply that the $\mathfrak{D}_{\alpha,\alpha}$ are not additive for the range $\alpha \in [1,\infty]$. In the following, we show that our framework allows extending this result to any monotone based on a quantum relative entropy. We remark that the non-additivity of $\mathfrak{D}_{1/2,1/2}$ implies that the fidelity of separability is not multiplicative for general states, which, to the best of our knowledge, was not known previously.

It is easy to prove that the $\mathfrak{D}_{\alpha,z}$ are sub-additive, i.e. $\mathfrak{D}_{\alpha,z}(\rho_1 \otimes \rho_2) \leq \mathfrak{D}_{\alpha,z}(\rho_1) + \mathfrak{D}_{\alpha,z}(\rho_1) $. We now prove that when $\rho_1=\rho_2=\rho_-$ where $\rho_-$ is the bipartite antisymmetric (Werner) state~\cite{vollbrecht2001entanglement,zhu2010additivity}, the previous inequality is strict. 
From Section~\ref{Werner states} we have that $\mathfrak{D}_{\alpha,z}(\rho_-) = 1$. We now consider the tensor product $\rho_- \otimes \rho_-$ of two antisymmetric states. We choose the (separable) ansatz~\cite{vollbrecht2001entanglement} $\tau _{-,-}= \frac{d+1}{2d} \rho_+ \otimes \rho_+ + \frac{d-1}{2d} \rho_- \otimes \rho_-$ which commutes with $\rho_- \otimes \rho_-$. We have
\begin{equation}
\Tr((\rho_- \otimes \rho_-)^\alpha \tau^{-\alpha}_{-,-}\sigma) \leq \left( \frac{d-1}{2d}\right)^{1-\alpha} = \Tr((\rho_-\otimes \rho_-)^\alpha \tau^{1-\alpha}_{-,-}) \,,
\end{equation}
for any $\sigma \in \SEP$. Here, we used that $\Lambda^2(\rho_-\otimes \rho_-) = \frac{d-1}{2d} \left(\frac{2}{d(d-1)}\right)^2$~\cite{zhu2010additivity}. Hence, the condition of Corollary~\ref{commuting} is satisfied and the ansatz is a solution. We then have for $(\alpha,z) \in \mathcal{D}$
\begin{align}
\label{single}
&\mathfrak{D}_{\alpha,z}(\rho_-)   = 1 \, , \; \text{and} \\
\label{tensor product}
& \mathfrak{D}_{\alpha,z}(\rho_- \otimes \rho_-)  = 1 - \log{\frac{(d-1)}{d}} \,.
\end{align}
Hence, additivity is violated for $d>2$. Note that, for $d\gg1$, we have  $\mathfrak{D}_{\alpha,z}(\rho_-) \sim \mathfrak{D}_{\alpha,z}(\rho_- \otimes \rho_-)$. For $d=2$, the additivity property is not violated since the state $\rho_-$ is a Bell diagonal state and, for what we discussed above, in this case, the monotones are additive. 
From Theorem 3 in~\cite{gour2020optimal} it follows that any monotone defined as $\min_{\sigma \in \SEP} \mathbb{D}(\rho \| \sigma)$, where $\mathbb{D}$ is a quantum relative entropy, satisfies $\mathfrak{D}_{\min}(\rho) \leq \min_{\sigma \in \SEP} \mathbb{D}(\rho \| \sigma) \leq \mathfrak{D}_{\max}(\rho)$. We recall that we have the limits~\eqref{limits monotone}. Hence, since the results in equation~\eqref{single} and ~\eqref{tensor product} do not depend on $\alpha$ and $z$, we have that additivity is violated for any monotone based on a quantum relative entropy. 

\begin{remark} 
In the literature, the minimization in~\eqref{problem} over the separable states is sometimes replaced by the minimization over the PPT states~\cite{rains1999bound,audenaert2001asymptotic}. The latter set is usually easier to characterize. We define
\begin{equation}
\label{problem sandwiched PPT}
\mathfrak{D}^{\PPT}_{\alpha,z}(\rho) := \min_{\sigma \in \text{PPT}} D_{\alpha,z}(\rho \| \sigma) \,.
\end{equation} 
We have for $(\alpha,z) \in \mathcal{D}$
\begin{align}
\label{PPT1}  
&\mathfrak{D}^\PPT_{\alpha,z}(\rho_-) =1\, , \; \text{and}\\
\label{PPT2}
& \mathfrak{D}^\PPT_{\alpha,z}(\rho_- \otimes \rho_-)  \leq 1 - \log{\frac{(d-1)}{d}} \,.
\end{align}
Indeed, following Section~\ref{Werner states}, to calculate the value of $\mathfrak{D}^ {\PPT}_{\alpha,z}$ for Werner states (and in particular for $\rho_-$), we can restrict the optimization to $\text{PPT}\cap W$ states. Since all entangled Werner states have negative partial transpose, the optimizations over PPT and SEP give the same value and hence~\eqref{PPT1} holds. The inequality in~\eqref{PPT2} follows from SEP $\subseteq$ PPT. Hence, also in this case, the monotones $\mathfrak{D}^\PPT_{\alpha,z}$ for $(\alpha,z) \in \mathcal{D}$  are not additive for $d>2$. As we discussed above,  this means that also any monotone $\min_{\sigma \in \text{PPT}} \mathbb{D}(\rho \| \sigma)$, where $\mathbb{D}$ is a quantum relative entropy, is not additive for general states. This result has already been established for $\alpha=z=1$ in~\cite{rains1998improved} where a different counterexample is provided. 
\end{remark}

\section{Some examples of states belonging to the additivity classes and their computation}
\label{Analytics}

In this section, we provide several examples of states that belong to the additivity classes that we discussed in the previous section. We also give a closed-form expression for the monotones $\mathfrak{D}_{\alpha,z}$ by extending some already known results. We summarize our results and our contribution in Table~\ref{states}. Our strategy consists of the following two steps. We first guess the optimizer, and then we check analytically if the conditions of Theorem~\ref{main Theorem} (or Corollary~\ref{commuting}) are satisfied. We remark that it is in general a hard problem to find the right guess for $\tau$ for general states. However, as we show later, the optimizer can be easily guessed in many situations as it usually lies at the boundary of the allowed separable region. In addition, numerical simulations could give insight on the structure of the optimizer.  We then need to solve the optimization problem $\Lambda^2(\Xi):=\max_{\sigma \in \SEP}\Tr(\sigma \Xi)$ where $\Xi$ is a positive operator. This optimization, as we see later, is usually an easier problem than the original one in Eq.~\eqref{problem}.  Indeed, because of the linearity of the trace, the optimization in $\Lambda^2(\Xi)$ can be restricted to the set of set pure product states and can be computed analytically for many states~\cite{zhu2010additivity}. Note that the optimization  $\Lambda^2(\Xi)$ could also be performed numerically by optimizing over the set of states with positive partial transpose instead of the separable states. Indeed, since the set of PPT states is larger than the set of separable ones~\cite{peres1996separability}, if the conditions of Theorem~\ref{main Theorem} are satisfied for any $\sigma \in \text{PPT}$, we can conclude that our (separable) guess is the right optimizer.

\setlength{\tabcolsep}{1.2pt}
\renewcommand{\arraystretch}{1.1}
\begin{table}
\begin{tabular}{l  l} 
\hline
 \multirow{2}{*}{\textbf{STATES}} &   \textbf{VALUE OF} $\mathfrak{D}_{\alpha,z}$ \\ \cmidrule{2-2}
& \textbf{REFERENCES} \\ 
\hline
\hline
\textbf{Bipartite pure}  & $  H_{\beta}(\vec{p}) \;, \quad \text{where}\;\; (1-\alpha)/z + 1/\beta = 1 $\\ \cmidrule{2-2}
$ |\rho(\vec{p}) \rangle = \sum_i \sqrt{p_i} |i,i \rangle$  & $z=\alpha=1$~\cite{vedral1998entanglement}, $z=\alpha, z=1$~\cite[Corollary 2]{zhu2017coherence} \, \\
& Other $(\alpha,z)$: Proposition~\ref{Bipartite pure}  \\
\hline
\hline
\textbf{Bell diagonal} &   $1-H_\alpha(\lambda_{\max},1-\lambda_{\max}) \;\;,   \text{if} \;\;  \lambda_{\max} \in \left[\frac{1}{2},1\right]$   \\ \cmidrule{2-2}
$\rho_{\text{BD}}(\vec{\lambda}) = \sum_{j=1}^{4}\lambda_j |\psi_j \rangle \! \langle \psi_j|$ & \multirow{1}{*}{$\alpha = z =1$~\cite{vedral1997quantifying,rains1998improved},$\,\infty$~\cite{zhu2010additivity}}\\
where $\{|\psi_j \rangle\}_{j=1}^4$ are the Bell states  & Other $(\alpha,z)$: Proposition \ref{pBD} \\
\addlinespace[0.1em]
\hline
\hline
\textbf{Generalized Werner}  &
 \multirow{2}{*}{ $                     1-H_\alpha(p,1-p) \, \,,\quad  p \in [0,1/2]  
                        $}\\
$ \rho_W(p)= p \frac{2}{d(d+1)}P^{\text{SYM}}_{AB} + (1-p)\frac{2}{d(d-1)}P^{\text{AS}}_{AB}$ & \\ \cmidrule{2-2}
$P^{\text{SYM}}_{AB}$ 
projector onto the symmetric subspace & $\alpha = z =1/2$~\cite{wei2003geometric},$\,1$~\cite{audenaert2001asymptotic,zhu2010additivity}\\
$P^{\text{AS}}_{AB}$ 
projector onto the antisymmetric subspace & Other $(\alpha,z)$: Proposition~\ref{pW}  \\
\addlinespace[0.1em]
\hline
\hline
\textbf{Isotropic} &  \multirow{2}{*}{$    \log{d}-H_\alpha\Bigg(\frac{1-F}{(d-1)^\frac{\alpha-1}{\alpha}},F\Bigg) \, \, F \in \left[\frac{1}{d},1 \right] $}\\
\vspace*{-\baselineskip} $ \rho_{iso}(F)=\frac{1-F}{d^2-1}(\mathds{1} -|\Phi^+\rangle \!\langle\Phi^+|) + F |\Phi^+\rangle \! \langle\Phi^+|$ & \\ \cmidrule{2-2}
where $  |\Phi^+ \rangle = \frac{1}{\sqrt{d}} \sum_i |ii\rangle$ & $\alpha = z= 1/2$~\cite{wei2003geometric},$\,1$~\cite{rains1998improved},$\,\infty$\cite{zhu2010additivity} \\ 
 & Other $(\alpha,z)$: Proposition~\ref{piso}\\
\hline
\hline
\textbf{Generalized Dicke} &   $-\log{\left(\frac{N!}{\prod_{j=0}^{d-1}k_j!}\prod_{j=0}^{d-1}\left(\frac{k_j}{N}\right)^{k_j} \right)}$   \\ \cmidrule{2-2}
$|S(N,\vec{k})\rangle = \frac{1}{\sqrt{C_{n,\vec{k}}}} \sum \limits_{ \text{P}} P |\underbrace{0, ..., 0}_{k_0},\underbrace{1, ..., 1}_{k_1}, ..., \underbrace{d-1, ..., d-1}_{k_d-1} \rangle$ & \multirow{2}{*}{$\alpha=z$~\cite{wei2003geometric,wei2008relative,hayashi2008entanglement}}\\
$\vec{k}=(k_0,...,k_{d-1}),\;  \sum_{j=0}^{d-1}k_j=N$ & Other $(\alpha,z)$: Proposition~\ref{gD}\\
\addlinespace[0.2em]
\hline
\hline
\textbf{Maximally correlated Bell diagonal} &   $\log{d}-H_\alpha(\vec{p})$\\ \cmidrule{2-2}
$\rho_{\text{MCBD}}(\vec{p}) = \sum_{k=0}^{d-1}p_k |\psi_k \rangle \! \langle \psi_k| $  & $\alpha = z = 1$~\cite{zhu2010additivity,rains1998improved}  \\
where $ |\psi_k \rangle = \frac{1}{\sqrt{d}} \sum_{j=0}^{d-1}e^{\frac{2 \pi i k}{d}j}|jj\rangle$ & Other $(\alpha,z)$: Proposition~\ref{pMCBD}\\
\addlinespace[0.3em]
\hline
\end{tabular}
\caption{The table contains the value of the monotones $\mathfrak{D}_{\alpha,z}$  for some states.  We denoted with $H_\alpha(\vec{p}) = \frac{1}{1-\alpha} \log{\left(\sum_i p_i^\alpha\right)}$ the $\alpha$-R\'enyi entropy of a vector $\vec{p}$. For Bell diagonal, Werner and isotropic states, we write only the range where the monotones are different from zero. For more details, see Section~\ref{Analytics}. The monotones $\mathfrak{D}_{\alpha,z}$ are additive when one state is among the ones listed in the table; the only exception is for Werner states, which provide a counterexample to the additivity of any entanglement monotone based on a quantum relative entropy (see Section~\ref{Counterexample}).}
\label{states}
\end{table}

\subsection{Bell Diagonal states (BD)}
We consider the Bell diagonal states 
\begin{equation}
\label{BD}
\rho_{\text{BD}}(\vec{\lambda})=\sum_{j=1}^{4}\lambda_j |\psi_j \rangle \! \langle \psi_j| \,,
\end{equation}
where $|\psi_j \rangle$ are the four Bell states. A Bell diagonal state is separable if $\lambda_i \in [0,1/2]$ for all $i$~\cite{horodecki1996information}.

\begin{proposition}
\label{pBD}
Let $(\alpha,z) \in \mathcal{D}$. For a Bell Diagonal state~\eqref{BD} we have \begin{align}
&\mathfrak{D}_{\alpha,z}(\rho_{\textup{BD}}(\vec{\lambda})) = \begin{cases}
0 & \text{if} \; \lambda_{\max} \in \left[0,\frac{1}{2}\right] \\
1-H_\alpha(\lambda_{\max},1-\lambda_{\max})  & \text{if} \;  \lambda_{\max} \in \left[\frac{1}{2},1\right]
\end{cases} \,.
\end{align}
\end{proposition}

\begin{proof}
A Bell diagonal state where all $\lambda_i \in [0,1/2]$  is separable and hence, in this range, we have $\min_{\sigma \in \SEP}D_{\alpha,z}(\rho_{\text{BD}}(\vec{\lambda}) \| \sigma ) = 0$. Without loss of generality, we assume that $\lambda_1 \geq 1/2$ and $\lambda_1 \geq \lambda_2 \geq \lambda_3 \geq \lambda_4$.  We take as ansatz $\tau_{\text{BD}} = \sum_{i=1}^{4} p_i |\psi_i \rangle \! \langle \psi_i|$ with $p_1=1/2$ and $p_i = \lambda_i/(2(1-\lambda_1)) \,\, \text{for} \,\, i \in \{2,3,4\}$. Using that $\Lambda^2(\rho_{\text{BD}}(\vec{\lambda})) = (\lambda_1+\lambda_2)/2$~\cite{zhu2010additivity} it is easy to check that $\Tr(\rho^\alpha_\text{BD}(\vec{\lambda}) \tau^{-\alpha}_\text{BD} \sigma) \leq 2^{\alpha -1} (\lambda_1^{\alpha} + (1-\lambda_1)^{\alpha}) = \Tr(\rho^\alpha_\text{BD}(\vec{\lambda}) \tau^{1-\alpha}_\text{BD})$ and the condition of Corollary~\ref{commuting} is satisfied.
Hence, our ansatz is an optimizer, and by explicit calculation, we obtain Proposition~\eqref{pBD}.
\end{proof}

The Bell diagonal states commute with their optimizer, and the product $\rho^\alpha_{\text{BD}}(\vec{\lambda}) \tau^{-\alpha}_{\text{BD}}$ is a (not normalized) Bell diagonal state. Bell diagonal states have positive entries in a product basis~\cite{zhu2010additivity}. Therefore, from Theorem~\ref{CC} it follows that the $\mathfrak{D}_{\alpha,z}$ is additive when one state is a Bell diagonal state.

\subsection{Werner states}
\label{Werner states}
A state $\rho \in \mathcal{S}(A_1 \otimes A_2)$ is called (generalized) Werner state if it is invariant under the projection~\cite{werner1989quantum,bennett1996purification,vollbrecht2001entanglement} ($UU$ group)
\begin{equation}
P: \rho \rightarrow \int dU (U \otimes U) \rho (U^\dagger \otimes U^\dagger) \,,
\end{equation} 
where $\d U$ is the invariant Haar measure. 
A Werner state (bipartite state in $d$ dimensions) can be written as
\begin{equation}
\label{W}
\rho_W(p) = p \frac{2}{d(d+1)}P^{\text{SYM}}_{AB} + (1-p)\frac{2}{d(d-1)}P^{\text{AS}}_{AB} \,,
\end{equation}
where $P^{\text{SYM}}_{AB}$ and $P^{\text{AS}}_{AB}$ are the projectors into the symmetric and antisymmetric subspace, respectively. The Werner state is separable for $p \geq 1/2$ and entangled for $p<1/2$~\cite{vollbrecht2001entanglement}. We define $\rho_+ := \rho_W(1)$ and $\rho_- := \rho_W(0)$. 

\begin{proposition}
\label{pW}
Let $(\alpha,z) \in \mathcal{D}$. For a generalized Werner state~\eqref{W} we have\begin{align}
&\mathfrak{D}_{\alpha,z}(\rho_{\text{W}}(p))  =
 \begin{cases}
 1-H_\alpha(p,1-p)  & \text{if} \;\;  p \in \left[0,\frac{1}{2}\right]\\
0 & \text{if} \; \; p \in \left[\frac{1}{2},1\right] 
\end{cases} \,.
\end{align}
\end{proposition}

\begin{proof}
We first calculate the quantity $\Lambda^2(\rho_W(p))=\max_{\sigma \in \SEP} \Tr(\sigma \rho_W(p))$. 
We first note that the optimization $\Lambda^2(\rho_W(p))$ can be restricted to separable Werner states. Indeed,
\begin{align}
\max \limits_{\sigma \in \SEP} \Tr(\sigma \rho_W(p)) & = \max \limits_{\sigma \in \SEP}  \Tr(\sigma \int dU (U \otimes U)\rho_W(p)(U^\dagger \otimes U^\dagger)) \\
& = \max \limits_{\sigma \in \SEP}  \Tr(\int dU(U^\dagger \otimes U^\dagger)\sigma (U \otimes U)\rho_W(p)) \\
& = \max \limits_{\sigma \in \SEP \cap \text{W}} \Tr(\sigma \rho_W(p)) \,,
\end{align}
where in the last step, we used that $P(\sigma) \in \SEP$ for any $\sigma \in \SEP$.
Then, by direct calculation 
\begin{equation}
\Tr(\rho_W(p')\rho_W(p'')) = p'p''\frac{2}{d(d+1)} + (1-p')(1-p'')\frac{2}{d(d-1)} \,,
\end{equation}
where we used that $\Tr(P^{SYM}_{AB})=d(d+1)/2$ and  $\Tr(P^{AS}_{AB})=d(d-1)/2$.
It is easy to see that the previous expression is maximized for $p'' = 1/2$ if $p' \leq \frac{d+1}{2d}$. We therefore obtain 
\begin{equation}
\Lambda^2(\rho_W(p'))= \frac{p'}{d(d+1)} + \frac{1-p'}{d(d-1)} \qquad \text{for} \; \;  p' \leq \frac{d+1}{2d} \,.
\end{equation}
We now calculate the values of $\mathfrak{D}_{\alpha,z}$. We choose as ansatz $\tau_W:=\rho_W (1/2)$ which commutes with any Werner state. Then $\rho_W(p)^\alpha \tau_W^{-\alpha}$ is (apart from a normalization factor) a Werner state with 
\begin{equation}
 p'= \left(1+\left(\frac{1-p}{p}\right)^\alpha \frac{d-1}{d+1}\right)^{-1} < (d+1)/(2d) \,,
\end{equation} 
where in the last inequality we assumed that the input state $\rho_W(p)$ is entangled, i.e. $p<1/2$. Then from the above, $\Tr(\sigma\rho_W(p)^\alpha \tau_W^{-\alpha}) \leq \Tr(\rho_W(p)^\alpha \tau_W^{1-\alpha})$ for $p<1/2$ for any $\sigma \in \SEP$. Hence, from Corollary~\ref{commuting} it follows that $\tau_W$ is an optimizer, and by explicit calculation, we obtain Proposition~\eqref{pW}.
\end{proof}

The Werner states do not satisfy the conditions for $\rho_1$ of Theorem~\ref{CC}, and they are not maximally correlated states. In Section~\ref{Counterexample} we show that they provide a counterexample to the additivity of the $\mathfrak{D}_{\alpha,z}$.

\subsection{Isotropic states}
\label{Isotropic states}
A state $\rho \in \mathcal{S}(A_1 \otimes A_2)$ is called an isotropic state if it is invariant under the projection ($UU^*$ group)~\cite{terhal2000entanglement,horodecki1999reduction}
\begin{equation}
P: \rho \rightarrow \int dU (U \otimes U^*) \rho (U^\dagger \otimes U^{*\dagger}) \,,
\end{equation} 
where $\d U$ is the invariant Haar measure. 
Isotropic states (bipartite states in $d$ dimensions) can be written as
\begin{equation}
\label{iso}
\rho_{iso}(F) = \frac{1-F}{d^2-1}(\mathds{1} -|\Phi^+\rangle\langle\Phi^+|) + F |\Phi^+\rangle\langle\Phi^+| \,,
\end{equation}
where $|\Phi^+ \rangle = 1/\sqrt{d} \sum_i |ii\rangle$ and $F \in [0,1]$. The isotropic state is separable for $F \leq 1/d$ and entangled for $F> 1/d$~\cite{horodecki1999reduction}. 

\begin{proposition}
\label{piso}
Let $(\alpha,z) \in \mathcal{D}$ . For an isotropic state~\eqref{iso} we have \begin{align}
& \mathfrak{D}_{\alpha,z}(\rho_{iso}(F) ) =
 \begin{cases}
0 & \text{if} \;\;  F \in \left[0,\frac{1}{d}\right]\\
\log{d}-H_\alpha\left(\frac{1-F}{(d-1)^\frac{\alpha-1}{\alpha}},F\right)  & \text{if} \; \; F \in \left[\frac{1}{d},1\right] 
\end{cases} \,.
\end{align}
\end{proposition}

\begin{proof}
For $F> 1/d$, we choose as ansatz $\tau_{iso}:=\rho_{iso} (1/d)$ which commutes with any isotropic state. Then $\rho_{iso}(F)^\alpha \tau_{iso}^{-\alpha}$ is (apart from a normalization factor) an isotropic state with \begin{equation}
F'= \left(1+\left(\frac{1-F}{F}\right)^\alpha \frac{d^2-1}{(d-1)^\alpha}\right)^{-1} > 1/d^2 \, ,
\end{equation} 
where in the last inequality we assumed that the input state $\rho_{iso}(F)$ is entangled, i.e. $F>1/d$.  We then use that $\Lambda^2(\rho_{iso}(F')) =  \frac{F'd+1}{d(d+1)}$ if $F'\geq 1/d^2$~\cite{zhu2010additivity}. Hence, $\Tr(\sigma \rho_{iso}(F)^\alpha \tau_{iso}^{-\alpha}) \leq  Q_\alpha(\rho_{iso}(F) \|\tau_{iso})$ for $F>1/d$ and any $\sigma \in \SEP$. From Corollary~\ref{commuting} it follows that $\tau_{iso}$ is an optimizer and by explicit calculation, we obtain Proposition~\eqref{piso}.
\end{proof}

The isotropic states commute with their optimizer, and the product $\rho^\alpha_{iso}(F) \tau^{-\alpha}_{iso}$ is a (not normalized) isotropic state. Isotropic states have positive entries in a product basis~\cite{zhu2010additivity}. Therefore, from Theorem~\ref{CC} it follows that the $\mathfrak{D}_{\alpha,z}$ are additive when one state is an isotropic state. 

\subsection{Bipartite pure states}
\label{Bipartite pure section}
The exact computation of the monotones $\mathfrak{D}_{\alpha,z}(|\rho\rangle)$ was already solved for bipartite pure states in~\cite{zhu2017coherence} for $\alpha=z$ and $z=1$. We extend this proof outside the former ranges and to a more general class of multipartite states, which includes the GHZ state.
\begin{proposition}
\label{Bipartite pure}
Let $(\alpha,z) \in \mathcal{D}$. For  a bipartite pure state $|\rho(\vec{p})\rangle = \sum \sqrt{p_i} |i,i \rangle$  we have
\begin{align}
& \mathfrak{D}_{\alpha,z}(|\rho(\vec{p})\rangle) = H_{\beta}(\vec{p}) \;, \quad \text{where}\;\; \frac{1-\alpha}{z} +  \frac{1}{\beta} = 1 \,.
\end{align}
\end{proposition}
\begin{proof}
Our ansatz for the optimzer is $\tau = \frac{1}{C}\sum_i p^{\beta}_i |i,i \rangle \! \langle i,i|$ where $C = \sum_i p_i^\beta$ is the normalization constant and $ \frac{1-\alpha}{z} + \frac{1}{\beta} = 1$. Here, $|i,i\rangle$ is the Schmidt basis of the pure state, i.e. the basis for which $|\rho(\vec{p})\rangle = \sum \sqrt{p_i} |i,i \rangle$. According to Proposition~\ref{MC marginal} we need to check the condition
\begin{equation}
\max_l \left(\frac{p_l^\beta}{C}\right)^{\frac{1-\alpha}{z}-1} \langle l,l|\chi_{\alpha,z}(\rho,\tau) |l,l\rangle \leq  Q_{\alpha,z}(\rho \|\tau) \,.
\end{equation} 
 We have that $\chi_{\alpha,z}(\rho,\tau) = \rho \langle \rho|\sigma^\frac{1-\alpha}{z}|\rho\rangle^{z-1}$, $Q_{\alpha,z}(\rho\|\tau) =  \langle \rho|\sigma^\frac{1-\alpha}{z}|\rho\rangle^{z}$ and $\langle \rho|\sigma^\frac{1-\alpha}{z}|\rho\rangle = C^\frac{1}{\beta}$. It is easy to check that the ansatz satisfies the above condition. Note that the maximum of the above quantity is achieved by $\sigma = |l,l \rangle$ for any $l$. Therefore, Proposition~\ref{Bipartite pure} holds. 
\end{proof}
 
Note that, in contrast with the previous cases, for bipartite pure states, the closest separable state depends on $\alpha$ and $z$. The bipartite pure states are maximally correlated states. Therefore, Theorem~\ref{TMC} implies that the $\mathfrak{D}_{\alpha,z}$ are additive when one state is a bipartite pure state. 

The same results can be extended to multipartite states of the form $|\rho\rangle = \sum_i \sqrt{p_i}|i,\dots,i\rangle$.  In particular, this class contains the GHZ state.

\begin{corollary}
Let $(\alpha,z) \in \mathcal{D}$ and $|\rho(\vec{p})\rangle = \sum \sqrt{p_i} |i,\dots,i \rangle \in \mathcal{S}(A_1 \otimes \cdots \otimes A_N)$. Then, we have
\begin{align}
& \mathfrak{D}_{\alpha,z}(|\rho(\vec{p})\rangle) = H_{\beta}(\vec{p}) \;, \quad \text{where}\;\;   \frac{1-\alpha}{z} + \frac{1}{\beta} = 1 \,.
\end{align}
\end{corollary}

\begin{proof}
The proof is analogous to the bipartite case.  In this case, $\tau = \frac{1}{C}\sum_i p^{\beta}_i \ketbra{i, \dots,i}{i,\dots,i}$ and we optimize over the pure separable states $\sigma = \sum_{k_1,\dots,k_N}c^{(1)}_{k_1}\dotsi c^{(N)}_{k_N} |k_1,\dots, k_N\rangle$. 
\end{proof}
The N-partite GHZ state where each subsystem has   dimension $d$ reads $|\text{GHZ} \rangle = \frac{1}{\sqrt{d}}\sum_{i=0}^{d-1}|i\rangle^{\otimes N}$. 
The corollary above implies that for the GHZ state $ \mathfrak{D}_{\alpha,z}(|\text{GHZ}\rangle) = \log{d}$. The latter result was already obtained for the cases $\alpha=z=1/2$~\cite{wei2003geometric},$1$~\cite{wei2004connections}, $\infty$~\cite{Regula}, and $\alpha=0,z=1$~\cite{Regula}. Since the result does not depend on $\alpha$ and $z$, the results for $\alpha=0,z=1$ and $\alpha=z=\infty$ already implies that $\min_{\sigma \in \SEP} \mathbb{D}(|\text{GHZ}\rangle \| \sigma) = \log{d}$ where $\mathbb{D}$ is any quantum relative entropy (see Theorem 3 in~\cite{gour2020optimal}).

\begin{remark}
The above computation together with the additivity result obtain in Section~\ref{Second class} immediately implies that for catalytic transformation of pure states $|\psi\rangle,|\phi\rangle$ with mixed catalysis $\nu$,  $|\psi\rangle \otimes \nu \fo |\phi\rangle \otimes \nu$, we must have $H_\beta(\vec{p}) \geq H_\beta(\vec{q})$ for all $\beta \in [1/2,\infty]$. Here, $\vec{p}$ and $\vec{q}$ are the Schmidt vector coefficients of $|\psi\rangle$ and $|\phi\rangle$, respectively.
This set of necessary conditions is a proper subset of the one for pure state catalysts~\cite{Klimesh}.
\end{remark}

\subsection{Generalized Dicke states}
Generalized Dicke states (or symmetric base states) are defined as~\cite{wei2003geometric,dicke1954coherence}
\begin{align}
\label{D}
&|S(N,\vec{k})\rangle = \frac{1}{\sqrt{C_{N,\vec{k}}}} \sum \limits_{\{P\}} P |\overbrace{0, \dots, 0}^{k_0},\overbrace{1, \dots, 1}^{k_1}, \dots, \overbrace{d-1, \dots, d-1}^{k_d-1} \rangle \,, \\
& \vec{k}=(k_0,\dots,k_{d-1}) \;, \quad  \text{and} \;\; \sum_{j=0}^{d-1}k_j=N \,.
\end{align}
Here $\{P\}$ denotes the set of all permutations and $C_{N,\vec{k}}=\frac{N!}{\prod_{j=0}^{d-1}k_j!}$ is the normalization factor.
\begin{proposition}
\label{gD}
Let $(\alpha,z) \in \mathcal{D}$. For a generalized Dicke state~\eqref{D} we have 
\begin{align}
\mathfrak{D}_{\alpha,z}(|S(N,\vec{k})\rangle) = -\log{\left(C_{N,\vec{k}}\prod_{j=0}^{d-1}\left(\frac{k_j}{N}\right)^{k_j} \right)} \,.
\end{align}
\end{proposition}

\begin{proof}
The separable ansatz of the problem is 
\begin{align}
\tau_{\text{GD}}=\int \limits_{[0,2\pi]^d} |\xi(\vec{\phi})\rangle \!\langle\xi(\vec{\phi})|\frac{\d \vec{\phi}}{(2\pi)^d} \; , \quad  \text{where} \quad  |\xi(\vec{\phi})\rangle = \left( \sum_{j=0}^{d-1} e^{i \phi_j}\sqrt{\frac{k_j}{N}}|j\rangle\right)^{\otimes n} \,.
\end{align}
Here, $\vec{\phi}=(\phi_0,\dots,\phi_{d-1})$ is a $d$-dimensional vector. The above state is a mixture of pure product states, and hence it is clearly separable. By explicit computation, we obtain
\begin{align}
\tau_{\text{GD}}=\sum_{\{k_j\}} C_{N,\vec{k}}\prod_{j=0}^{d-1}\left(\frac{k_j}{N}\right)^{k_j} |S(N,\vec{k})\rangle \langle S(N,\vec{k})| \,,
\end{align}
where the sum is taken over all combinations of nonnegative integer indices $k_0$ through $k_{d-1}$ such that $\sum_{j=0}^{d-1}k_j=N$.
The above ansatz commutes with the input state. Using that $\Lambda^2_{\max}(|S(N,\vec{k})\rangle) = C_{N,\vec{k}}\prod_{j=0}^{d-1}\left(\frac{k_j}{N}\right)^{k_j}$~\cite{wei2003geometric,wei2008relative,zhu2010additivity} and Corollary~\ref{commuting} we get Proposition~\ref{gD}.
\end{proof}

The generalized Dicke states commute with their optimizer, and the product $ (|S(N,\vec{k})\rangle \langle S(N,\vec{k})|)^\alpha $ $\times \tau^{-\alpha}_{\text{GD}} \propto |S(N,\vec{k})\rangle \langle S(N,\vec{k})|$ is a (not normalized) generalized Dicke state. The latter states have positive entries in a product basis~\cite{zhu2010additivity}. Therefore, from Theorem~\ref{CC} it follows that the $\mathfrak{D}_{\alpha,z}$ are additive when one state is a generalized Dicke state. 

\subsection{Maximally correlated Bell diagonal states (MCBD)}
A closed form for the maximally correlated states is known only for the entanglement monotone based on the Petz divergence~\cite{zhu2017coherence}.  We could not find a closed form for the $\alpha$-$z$ R\'enyi divergences of entanglement for maximally correlated states for all $(\alpha,z) \in \mathcal{D}$. Nevertheless, in Section~\ref{Second class}, we could show that the $\mathfrak{D}_{\alpha,z}$ are additive whenever one state is maximally correlated. In the following, to obtain a closed form, we add another constraint, i.e. we require the Bell diagonal condition. We consider the maximally correlated Bell diagonal states (MCBD)~\cite{zhu2010additivity}
\begin{equation}
\label{MCBD}
\rho_{\text{MCBD}}(\vec{p})=\sum_{k=0}^{d-1}p_k \ketbra{\psi_k}{\psi_k}  \;, \quad \text{where}\;\;  |\psi_k \rangle = \frac{1}{\sqrt{d}} \sum_{j=0}^{d-1}e^{\frac{2 \pi i k}{d}j}|jj\rangle \,.
\end{equation}
\begin{proposition}
\label{pMCBD}
Let $(\alpha,z) \in \mathcal{D}$. For a maximally correlated Bell diagonal state~\eqref{MCBD} we have 
\begin{align}
\mathfrak{D}_{\alpha,z}(\rho_{\textup{MCBD}}(\vec{p}))  = \log{d}-H_\alpha(\vec{p}) \,.
\end{align}
\end{proposition}

\begin{proof}
We take as ansatz $\tau_{\text{MCBD}} = \sum_{k=0}^{d-1} \frac{1}{d} \ketbra{\psi_k}{ \psi_k}$ which satisfies $[\rho_{\text{MCBD}}(\vec{p}),\tau_{\text{MCBD}}]=0$. Using that $\Lambda^2(\rho_{\text{MCBD}}(\vec{p}))=1/d$~\cite{zhu2010additivity} we have that $\Tr(\rho^\alpha_{\text{MCBD}}(\vec{p}) \tau^{-\alpha}_{\text{MCBD}} \sigma) \leq \sum p_k^{\alpha} d^{\alpha-1} $ = $ \Tr(\rho^\alpha_{\text{MCBD}}(\vec{p}) \tau^{1-\alpha}_{\text{MCBD}})$ for any $\sigma \in \SEP$ and the condition of Corollary~\ref{commuting} is satisfied. Hence, our ansatz is an optimizer, and by explicit calculation, we obtain Proposition~\eqref{pMCBD}.
\end{proof}

The maximally correlated Bell diagonal states are maximally correlated states. Therefore, Theorem~\ref{TMC} implies that the $\mathfrak{D}_{\alpha,z}$ are additive when one state is a maximally correlated Bell diagonal state.

\section*{ACKNOWLEDGMENTS}
We thank Paolo Perinotti, Erkka Happasalo, and Ryuji Takagi for discussion. Part of the work was conducted while R.R was hosted by the QUIT Group at the University of Pavia and M.T. was visiting the Pauli Centre for Theoretical studies at ETH Zurich. This research is supported by the National Research Foundation, Prime Minister’s
Office, Singapore and the Ministry of Education, Singapore under the Research Centres of Excellence programme.
MT is also supported in part by NUS startup grants (R-263-000-E32-133 and R-263-000-E32-731).

\section*{Statements and Declarations}
The authors declare no competing interests. 

\section*{Data availability}
No data sets were generated during this study. 

\bibliographystyle{ultimate}
\bibliography{my}

\newpage
 
\appendix
\section*{Appendices}
\addcontentsline{toc}{section}{Appendices}
\renewcommand{\thesubsection}{\Alph{subsection}}

\subsection{Lower semicontinuity and its consequences}
\label{lower-semicontinuous}

In this appendix, we show that the $\alpha$-$z$ R\'enyi entropy of entanglement are lower semicontinuous and hence the infimum of $\mathfrak{D}_{\alpha,z}$ is always achieved in the set of separable states. Moreover, we show that the monotones $\mathfrak{D}$ and $\mathfrak{D}_{\max}$ can be obtained by taking the limits $\alpha \rightarrow 1$ ($z \neq 0$) and $\alpha=z \rightarrow \infty$ of $\mathfrak{D}_{\alpha,z}$, respectively. More in general, we prove that the latter property holds also when the set of separable states is replaced by any set of free states. All our considerations are in finite dimensions.  We follow similar arguments to the ones given in~\cite{mosonyi2011quantum} and~\cite{lami2021attainability}. 

It is a straightforward fact that the $D_{\alpha,z}$ are monotonically decreasing in the second argument in the range where they satisfy the data-processing inequality.
\begin{lemma}
\label{decreasing}
Let $(\alpha,z) \in \mathcal{D}$, $\rho\in \qstate $, and $\sigma,\sigma' \in \mathcal{P}(A)$. Then if $\sigma \leq \sigma' $ we have
\begin{align}
&D_{\alpha,z}(\rho \| \sigma) \geq D_{\alpha,z}(\rho \| \sigma')  \,.
\end{align}
\end{lemma}
\begin{proof}
We split the proof for the region $(\alpha,z) \in \mathcal{D}$ in the range $\alpha <1$, $\alpha>1$ and $\alpha=1$. 

Let $\alpha <1$. We then have 
\begin{align}
 \sigma \leq \sigma'
\implies  \sigma^\frac{1-\alpha}{z} \leq {\sigma'}^\frac{1-\alpha}{z} 
\implies  \Tr[(\rho^\frac{\alpha}{2z} \sigma^\frac{1-\alpha}{z}  \rho^\frac{\alpha}{2z})^z] \leq \Tr[(\rho^\frac{\alpha}{2z} {\sigma' }^\frac{1-\alpha}{z}  \rho^\frac{\alpha}{2z})^z]  
\label{Impl2} \,.
\end{align}
In the first implication, we used that $(1-\alpha)/z \in (0,1]$ and that the power is operator monotone in the range $(0,1]$. In the last implication, we used that the trace functional $M \rightarrow \Tr(f(M))$ inherits the monotonicity from $f$~(see e.g.~\cite{carlen2010trace}). Therefore, in this range, we get $D_{\alpha,z}(\rho \| \sigma) \geq D_{\alpha,z}(\rho \| \sigma')$. 

The proof for $\alpha >1$ follows in the same way by noticing that $(1-\alpha)/z \in [-1,0)$ and that the power in this range is operator antimonotone. 

The property for $\alpha=1$ follows by taking the limit $\alpha\rightarrow 1$. 
\end{proof}

\begin{lemma}
The $\alpha$-$z$ R\'enyi relative entropy
$
(\rho, \sigma) \rightarrow D_{\alpha,z}(\rho \| \sigma)
$
is lower semicontinuous for all $(\alpha,z)\in \mathcal{D}$. 
\end{lemma}

\begin{proof}
For any $\varepsilon>0$, we have the operator inequality $\sigma \leq \sigma + \varepsilon \mathds{1}$ where $\varepsilon>0$.
From the previous lemma we have $ D_{\alpha,z}(\rho \| \sigma) \geq D_{\alpha,z}(\rho \| \sigma+\varepsilon \mathds{1})$ and hence we can write
\begin{equation}
\label{sup}
D_{\alpha,z}(\rho \| \sigma) = \sup \limits_{\varepsilon>0} D_{\alpha,z}(\rho \| \sigma+\varepsilon \mathds{1}) \,.
\end{equation} 
For any fixed $\varepsilon>0$, the functions $(\rho, \sigma) \rightarrow D_{\alpha,z}(\rho \| \sigma+\varepsilon \mathds{1})$ are continuous since the second argument has full support. Since the pointwise supremum of continuous functions is lower semicontinuous, the function  $(\rho, \sigma) \rightarrow D_{\alpha,z}(\rho \| \sigma)$ is lower semicontinuous. 
\end{proof}
A direct consequence of the above lemma is that $\inf_{\sigma \in \mathcal{F}}D_{\alpha,z}(\rho \| \sigma) = \min_{\sigma \in \mathcal{F}}D_{\alpha,z}(\rho \| \sigma)$, i.e. the infimum is always achieved. Indeed, since the set of quantum states is compact (with respect, for example, to the trace norm topology), the set $\mathcal{F}$ is also compact being a closed subset of a compact set. The statement then follows from the fact that a lower semicontinuous function on a compact set has a minimum on the compact set.

Moreover, lower semicontinuity implies that
\begin{equation}
\lim \limits_{\alpha \rightarrow \infty}\min \limits_{\sigma \in \mathcal{F}}D_{\alpha,\alpha}(\rho \| \sigma) = \sup \limits_{\alpha > 1} \inf \limits_{\sigma \in \mathcal{F}}D_{\alpha,\alpha}(\rho \| \sigma) = \inf \limits_{\sigma \in \mathcal{F}}   \sup \limits_{\alpha > 1}D_{\alpha,\alpha}(\rho \| \sigma) = \min \limits_{\sigma \in \mathcal{F}}  \lim \limits_{\alpha \rightarrow \infty}D_{\alpha,\alpha}(\rho \| \sigma) \,.
\end{equation}
i.e. we can exchange the limit with the minimum. 
The first equality follows from the fact that the sandwiched R\'enyi divergence is monotone in $\alpha$ and hence we can replace the limit with the supremum over $\alpha>1$. The second equality follows from the minimax theorem in~\cite[Lemma II.1]{mosonyi2022some} since, as we discussed above, the $D_{\alpha,\alpha}$ is lower semicontinuous in the second argument. Moreover, by the minimax theorem, in the last equality, the infimum over the free states can be replaced with the minimum.
Similarly, we can replace the following limits with the infimum or supremum to obtain 
\begin{enumerate}
\item $\lim \limits_{z \rightarrow \infty}\min \limits_{\sigma \in \mathcal{F}}D_{\alpha,z}(\rho \| \sigma) = \min \limits_{\sigma \in \mathcal{F}}  \lim \limits_{z \rightarrow \infty}D_{\alpha,z}(\rho \| \sigma)\; $ for $\; \alpha<1$.
\item $\lim \limits_{z \rightarrow \alpha-1}\min \limits_{\sigma \in \mathcal{F}}D_{\alpha,z}(\rho \| \sigma) = \min \limits_{\sigma \in \mathcal{F}}  \lim \limits_{z \rightarrow \alpha-1}D_{\alpha,z}(\rho \| \sigma)\; $ for $\; \alpha > 1$.
\item $\lim \limits_{z \rightarrow 1-\alpha}\min \limits_{\sigma \in \mathcal{F}}D_{\alpha,z}(\rho \| \sigma) = \min \limits_{\sigma \in \mathcal{F}}  \lim \limits_{z \rightarrow 1-\alpha}D_{\alpha,z}(\rho \| \sigma)\; $ for $\; \alpha < 1$.
\item $\lim \limits_{\alpha \rightarrow 0}\min \limits_{\sigma \in \mathcal{F}}D_{\alpha,1}(\rho \| \sigma) = \min \limits_{\sigma \in \mathcal{F}}  \lim \limits_{\alpha \rightarrow 0}D_{\alpha,1}(\rho \| \sigma)$. 
\item $\lim \limits_{\alpha \rightarrow 1}\min \limits_{\sigma \in \mathcal{F}}D_{\alpha,1}(\rho \| \sigma) = \min \limits_{\sigma \in \mathcal{F}}  \lim \limits_{\alpha \rightarrow 1}D_{\alpha,1}(\rho \| \sigma) $ .
\end{enumerate}
The first two equalities follow from the minimax theorem~\cite[Lemma II.1]{mosonyi2022some} and the fact that the function $z \mapsto D_{\alpha,z}$ is monotonically increasing for $\alpha<1$ and monotonically decreasing for $\alpha>1$~\cite{lin2015investigating}. The third equality follows from the monotonicity just mentioned. The last two equalities follow from the monotonicity in $\alpha$ of the Petz R\'enyi divergence~\cite{Tomamichel}.

\subsection{Relationship between conditional entropies, entanglement monotones, and coherence monotones based on $\alpha$-$z$ R\'enyi divergences}
\label{relations}
In this appendix, we derive some relationships between the conditional entropy and the entanglement monotone on the $\alpha$-$z$ R\'enyi divergence. Moreover, we derive some constraints on the structure of an optimizer of $\mathfrak{D}_{\alpha,z}$ for maximally correlated states, which constitutes a fundamental ingredient to establish the additivity result of Theorem~\ref{TMC}. 

Let $\rho_{AB} \in \mathcal{S}(A B)$ a bipartite state. We define the $\alpha$-$z$ \textit{R\'enyi conditional entropy} as
\begin{equation}
H^\uparrow_{\alpha,z}(A|B)_{\rho_{AB}} := -\min \limits_{\sigma_B \in \mathcal{S}(B)}D_{\alpha,z}(\rho_{AB}\| I_{A} \otimes \sigma_{B}) \,.
\end{equation}
We often omit the subscript $AB$ when it is clear from the context.

We now show that for maximally correlated states, the $\alpha$-$z$ R\'enyi conditional entropies are equal to the $\alpha$-$z$ R\'enyi relative entropies of entanglement (up to a minus sign).  We follow similar arguments to the ones given in~\cite{zhu2017coherence}.
\begin{theorem}
\label{equivalence}
Let $\rho \in \mathcal{S}(AB)$ be a maximally correlated state and $(\alpha,z) \in \mathcal{D}$. Then we have
\begin{equation}
\label{lineq1}
\mathfrak{D}_{\alpha,z}(\rho) = -H^\uparrow_{\alpha,z}(A|B)_{\rho} \,.  
\end{equation}
\end{theorem}
\begin{proof}
We define the CPTP map $\Lambda$ by $\Lambda(\rho_{AB}) := P_{AB}\rho_{AB} P_{AB} + (1-P_{AB}) \rho_{AB} (1-P_{AB})$ where $P_{AB} = \sum_j|jj\rangle \! \langle jj|_{AB}$. We have for a maximally correlated state $\rho$
\begin{align}
\label{first inequality}
\mathfrak{D}_{\alpha,z}(\rho_{AB}) &\geq \min \limits_{\sigma_B \in \mathcal{S}(B)} D_{\alpha,z}(\rho_{AB} \| I_{A} \otimes \sigma_B) \\
& \geq  \min \limits_{\sigma_B \in \mathcal{S}(B)} D_{\alpha,z}(\rho_{AB} \| P_{AB}(I_{A} \otimes \sigma_B)P_{AB} + (1-P_{AB})(I_{A} \otimes \sigma_B)(1-P_{AB}))\\
\label{teq}
& = \min \limits_{\sigma_B \in \mathcal{S}(B)} D_{\alpha,z}(\rho_{AB} \| P_{AB}(I_{A} \otimes \sigma_B)P_{AB})\\
& \geq \min \limits_{\sigma_{AB} \in \mathcal{T}_{\rho_{AB}}}D_{\alpha,z}(\rho_{AB} \| \sigma_{AB})\,,
\end{align}
where $\mathcal{T}_{\rho_{AB}}$ is the set of all bipartite separable states of the form $\sigma_{AB} = \sum_i s_i |i,i\rangle \! \langle i,i|_{AB}$. Here, $|i,i\rangle$ is the basis such that $\rho = \sum_{jk} \rho_{jk} |j,j\rangle \! \langle k,k|$.
The first inequality comes from the fact that for any separable state $\sigma_{AB} \leq I_A \otimes \sigma_B$~\cite{horodecki1999reduction} and Lemma~\ref{decreasing}.
In the second inequality, we used the data-processing inequality and that a maximally correlated state is invariant under the projection $P_{AB}$, i.e. $\Lambda(\rho_{AB}) = \rho_{AB}$. In~\eqref{teq} we used that $(1-P_{AB})(I_A \otimes \sigma_B)(1-P_{AB})$ is orthogonal to $\rho_{AB}$ and hence it does not contribute to the value of $D_{\alpha,z}$. The last inequality is a consequence of the fact that $P_{AB}(I_A \otimes \sigma_B)P_{AB} = \sum_i(\sigma_B)_{i,i}|i,i\rangle \! \langle i,i|_{AB}$ belongs to $\mathcal{T}_{\rho_{AB}}$. Moreover, since $\mathcal{T}_{\rho_{AB}} \subseteq \SEP$, we have $\min_{\sigma_{AB} \in \mathcal{T}_{\rho_{AB}}}D_{\alpha,z}(\rho_{AB} \| \sigma_{AB}) \geq \mathfrak{D}_{\alpha,z}(\rho_{AB})$ and hence $\min_{\sigma_{AB} \in \mathcal{T}_{\rho_{AB}}} D_{\alpha,z}(\rho_{AB} \| \sigma_{AB}) = \mathfrak{D}_{\alpha,z}(\rho_{AB})$. The latter condition implies that the inequality~\eqref{first inequality} is equality and equation~\eqref{lineq1} holds.
\end{proof}
The latter proof implies that for a maximally correlated state, there always exists an optimizer of $\mathfrak{D}_{\alpha,z}$ in the set $ \mathcal{T}_{\rho}$. 

\begin{remark}
In general, the value of the coherence monotones depends on the specific choice of the coherence basis. Note that, if we set the coherence basis equal to the basis $\{ |i,j\rangle \}$ for which the maximally correlated state reads $\rho = \sum_{jk}\rho_{jk}|j,j\rangle \! \langle k,k|$, following the same arguments as above, we have that $\mathfrak{D}_{\alpha,z}(\rho) = \mathfrak{D}^{\mathcal{I}}_{\alpha,z}(\rho)$. Moreover, because of the invariance of the underlying relative entropy under the isometry $|i\rangle \rightarrow |i,i\rangle$, we have that the latter quantity is equal to $\mathfrak{D}^{\mathcal{I}}_{\alpha,z} \big(\sum_{i,j}\rho_{jk}|j\rangle \! \langle k|\big)$, where now the minimization runs over the states diagonal in the marginal basis $\{ |i\rangle \}$.
\end{remark}

Proposition~\ref{Bipartite pure} together with Theorem~\ref{equivalence} gives
\begin{lemma}
Let $(\alpha,z) \in \mathcal{D}$ and $\rho_{AB} \in \mathcal{S}(AB)$ a bipartite pure state. Then we have
\begin{equation}
H^\uparrow_{\alpha,z}(A|B)_{\rho} + H_\beta(A)_{\rho} = 0  \;, \quad \text{where}\;\;    \frac{1-\alpha}{z} + \frac{1}{\beta}  = 1 \,.
\end{equation}
Here, $H_{\alpha}(A)_{\rho}:= \frac{1}{1-\alpha}\log{\Tr(\rho_A^\alpha)}$.
\end{lemma}
The quantity $H_{\alpha}(A)_{\rho}$ is the quantum R\'enyi entropy of order $\alpha$ of the marginal $\rho_A := \Tr_B(\rho_{AB})$.

\end{document}